\RequirePackage[l2tabu,orthodox]{nag}
\documentclass[11pt]{article}
\pdfoutput=1 

\usepackage{tcolorbox}
\usepackage{fullpage}
\usepackage{amsmath,amssymb,amsthm,amsfonts,amstext}
\usepackage{latexsym,bbm,xspace,float,mathtools}
\usepackage[margin=1in]{geometry}
\usepackage[backref, colorlinks,citecolor=blue,linkcolor=magenta,bookmarks=true]{hyperref}
\usepackage{tikz}
\usepackage{youngtab}
\usepackage[utf8]{inputenc}
\usepackage{cite}
\usepackage{verbatim}
\usepackage{bbold}
\usepackage{csquotes}
\usepackage{subfig}
\usepackage[normalem]{ulem}

\usepackage{thmtools}
\usepackage{thm-restate}

\usepackage{graphicx}
\usepackage{enumerate}
\usepackage{etoolbox}
\usetikzlibrary{decorations, patterns}
\newcommand{\E}{\ensuremath{\mathbb{E}}}

\newcommand{\BWT}{\mathsf{BWT}}
\newcommand{\MTF}{\mathsf{MTF}}

\newcommand{\RLE}{\mathsf{RLE}}

\newcommand{\HRLX}{\texttt{RLX}}

\newcommand{\cA}{\mathcal{A}}

\newcommand{\cD}{\mathcal{D}}
\newcommand{\cE}{\mathcal{E}}

\newcommand{\cI}{\mathcal{I}}

\newcommand{\cL}{\mathcal{L}}
\newcommand{\cM}{\mathcal{M}}
\newcommand{\cN}{\mathcal{N}}

\newcommand{\cP}{\mathcal{P}}

\newcommand{\cR}{\mathcal{R}}
\newcommand{\cS}{\mathcal{S}}
\newcommand{\cT}{\mathcal{T}}
\newcommand{\cU}{\mathcal{U}}

\newcommand{\Pat}{P\v{a}tra\c{s}cu}

\newcommand{\Var}{\operatorname{{\bf Var}}}

\renewcommand{\Pr}{\operatorname{{\bf Pr}}}

\newcommand{\poly}{\mathrm{poly}}

\newcommand{\N}{\mathbbm N}

\newcommand{\eps}{\varepsilon}
\renewcommand{\epsilon}{\eps}

\newcommand{\wh}[1]{\widehat{#1}}

\newcommand{\zo}{\{0,1\}}

\renewcommand{\hat}{\wh}

\newcommand{\kap}{\kappa}

\newcommand{\mnote}[1]{ \marginpar{\tiny\bf
		\begin{minipage}[t]{0.5in}
			\raggedright #1
\end{minipage}}}

\newtheorem{definition}{Definition}
\newtheorem{theorem}{Theorem}
\newtheorem{lemma}{Lemma}
\newtheorem{proposition}{Proposition}

\newtheorem{claim}{Claim}
\newtheorem{fact}{Fact}

\newtheoremstyle{restate}{}{}{\itshape}{}{\bfseries}{~(restate).}{.5em}{\thmnote{#3}}
\theoremstyle{restate}

\def\colorful{1}

\ifnum\colorful=1

\newcommand{\blue}[1]{{{\color{blue}#1}}}

\fi
\ifnum\colorful=0

\newcommand{\blue}[1]{{{#1}}}

\fi

\title{Local Decodability of the Burrows-Wheeler Transform}
\date{}
\author{
	Sandip Sinha\thanks{Supported by NSF awards CCF-1563155, CCF-1420349, CCF-1617955, 
		CCF-1740833, CCF-1421161, CCF-1714818 and Simons Foundation (\#491119).}\\
	Columbia University\\
	sandip@cs.columbia.edu
	\and
	Omri Weinstein\\
	Columbia University\\
	omri@cs.columbia.edu
}

\begin{document}

\maketitle


\thispagestyle{empty}

\begin{abstract}
The Burrows-Wheeler Transform (BWT) is among the 
most influential discoveries in 
text compression and DNA storage. 
It is a reversible preprocessing step that rearranges 
an $n$-letter string into 
runs of identical characters (by exploiting context regularities), 
resulting in highly compressible strings, and is the basis of the 
\texttt{bzip}  compression program. 
Alas, the decoding process of BWT is inherently sequential and requires $\Omega(n)$ 
time even to retrieve a \emph{single} character. 

We study the succinct data structure problem of locally decoding short substrings of a given text under its 
\emph{compressed} BWT, 
i.e., with small additive redundancy $r$ over the \emph{Move-To-Front}  (\texttt{bzip}) compression.  
The celebrated BWT-based FM-index  
(FOCS '00), 
as well as other related literature, yield 
a trade-off of $r=\tilde{O}(n/\sqrt{t})$ bits, when a single 
character is to be decoded in $O(t)$ time. 
We give a near-quadratic improvement $r=\tilde{O}(n\lg(t)/t)$. 
As a by-product, we obtain an \emph{exponential} (in $t$) improvement on the redundancy of the FM-index for 
counting pattern-matches on compressed text. 
In the 
interesting regime where the text compresses to $n^{1-o(1)}$ bits, these results provide an 
$\exp(t)$ \emph{overall} space reduction.  
For the local decoding problem of BWT, we also prove an $\Omega(n/t^2)$ cell-probe lower bound 
for ``symmetric" data structures.    

We achieve our main result by designing a compressed partial-sums (Rank) data structure over BWT.     
The key component is a \emph{locally-decodable} Move-to-Front (MTF) code: 
with only $O(1)$ extra bits per block of length $n^{\Omega(1)}$, 
the decoding time of a single character
can be decreased from $\Omega(n)$ to $O(\lg n)$. 
This result is of independent interest in algorithmic information theory. 
\end{abstract}

\newpage

\setcounter{page}{1}

\section{Introduction} \label{sec_intro}

Exploiting  text regularities for data compression is an 
enterprise that received a tremendous amount of attention in the 
past several decades, driven by large-scale digital storage. 
Motivated by this question, in 1994 Burrows and Wheeler  \cite{BW} proposed a 
preprocessing step 
that `tends' to rearrange strings  into more `compressible form': 
Given $x\in \Sigma^n$, generate the $n\times n$ matrix whose rows are all \emph{cyclic shifts} of $x$, sort these rows 
lexicographically, and output the last \emph{column} $L$  
of the matrix. The string $L$ is called the Burrows-Wheeler Transform (BWT) of $x$.  
Note that $L$ is a \emph{permutation} 
of the original string $x$, as each row of the sorted matrix is still a unique cyclic shift. 
The main observation is that, in the permuted string $L$, 
characters with identical \emph{context}\footnote{The $k$-\emph{context} of a character $x_i$ in $x$ 
is the $k$ consecutive characters that precede it.} 
appear consecutively, hence if 
individual characters in the original text $x$ tend to be predicted by a reasonably small context 
(e.g., as in English texts or bioinformatics \cite{Adjeroh08book}), 
then the string $L := \BWT(x)$ 
will exhibit local similarity. That is, identical symbols will tend to recur at \emph{close vicinity}. 
This property suggests a natural way to compress $L$, using a relative 
\emph{recency coding} method, whereby each symbol in $L$ is replaced by the \emph{number of distinct symbols that appeared since its last occurrence}. Indeed, since in $L$ symbols have local similarity, i.e.,  
tend to recur at close vicinity, we expect the output string to consist mainly of \emph{small integers} (and in particular, $0$-runs) and hence be much cheaper to describe. 
This relative encoding, known as the \emph{Move-to-Front} transform  \cite{Bentley}, 
followed by 
run-length coding of $0$-runs and an arithmetic coding stage  (henceforth denoted $\HRLX(L)$),  
forms the basis for the widespread \texttt{bzip2} program \cite{bzip2}. 
The $\HRLX$ (``bzip") compression benchmark was justified both theoretically and empirically 
\cite{Manzini99,Kaplan, Effros02}, where among other properties, 
it was shown to approximately converge to \emph{any finite-order} empirical entropy $H_k$  
(see Section \ref{sec_RLX_comparison} for the 
formal definitions and Appendix \ref{sec_app_RLX} for comparative analysis against other compressors).  


Remarkably, 
the Burrows-Wheeler transform of a string 
is \emph{invertible}. 
The crux of the decoding process is the fact that the transform preserves the \emph{order} of occurrences (a.k.a \emph{rank}) of 
identical symbols in both the first column and last column ($L$) of the BWT matrix. 
This crucial fact 
facilitates an iterative decoding process, whereby, given the decoded position of $x_{i+1}$ in $L$, 
one can decode the 
previous character $x_{i}$ using $O(|\Sigma|)$ \textsc{Rank}\footnote{The \textsc{Rank} of a character 
$x_i \in x$ is the number of occurrences of identical symbols in $x$ that precede it : $\textsc{Rank}(x,i) := |\{ j \leq i \; :  x_j = x_i\}|$.  
See also Section \ref{sec_prelims}.}
queries \emph{to $L$} (see Section \ref{sec_decoding_BWT} below for the formal decoding algorithm). 
Alas, this decoding process is inherently \emph{sequential}, and therefore requires $\Omega(n)$ time to decode 
even a single coordinate of $x$ \cite{Kaplan_lower_bound,FM_lightweight_suffix_array}. 
In fact, no sub-linear decoding algorithm for $x_i$ is known even if $L$ is stored in uncompressed form. 

This is an obvious drawback of the Burrows-Wheeler transform, as many storage applications, such as 
genetic sequencing and alignment, need local searching capabilities inside the compressed database \cite{NavarroSurvey07}. 
For example, if $x^{1}, x^{2}, \ldots, x^{m} \in \Sigma^n $ is a collection of $m$  
separate files with very similar contexts, e.g., DNA sequences, 
then we expect  $|\HRLX(L_{x^{1}\circ x^{2} \circ\ldots \circ x^{m}})| \ll \sum_{j=1}^m |\HRLX(L_{x^{j}})|$, 
but jointly compressing the files would have a major drawback -- When the application needs to retrieve only a single file $x^{j}$, 
it will need to spend $\Omega(n \cdot m)$ I/Os (instead of $O(n)$) to invert $\BWT(x^{1}\circ x^{2} \circ\ldots \circ x^{m})$. 
The main question we address in this paper is, whether small \emph{additive} space redundancy (over the compressed 
BWT  string 
($|\HRLX(\BWT(x))|$) can be used to speed up the decoding time of a single character (or a short substring) of the original text, in the word-RAM model: 
\begin{quote} \label{problem_1}
\bf Problem 1.\rm
\emph{ What is the least amount of space redundancy $r=r(t)$ 
needed beyond the compressed $\BWT(x)$ 
string, so that each coordinate $x_i$ can be decoded in time $t$ ? } 
\end{quote}

Since $|\HRLX(\BWT(x))|$ approaches \emph{any} finite-order empirical entropy $H_k(x)$ (see Section \ref{sec_RLX_comparison}), 
this data structure problem can be viewed as the \emph{succinct dictionary}\footnote{For a string $x \in \Sigma^n$ and an index $i \in [n]$, \textsc{Dictionary}$(x, i)$ returns $x_i$, the $i^{th}$ character in $x$.}  
problem under \emph{infinite-order} entropy 
space benchmark. A long line of work has been devoted to succinct dictionaries in the \emph{context-free} regime, i.e., when 
the information theoretic space benchmark is the \emph{zeroth-order} empirical entropy $H_0(x) 
 := \sum_{c\in \Sigma} n_c \lg \frac{n}{n_c}$ of marginal frequencies (e.g., \cite{Pat,Pagh,GGGRR} to mention a few). 
However, as discussed below, 
much less is known about the best possible trade-off under higher-order entropy benchmarks which take context 
(i.e., correlation between $x_i$'s) into account.  

A related but incomparable data structure problem to Problem $1$, 
is that of compressed pattern-matching in strings, 
a.k.a \emph{full-text indexing}, where the goal is to succinctly represent a text in \emph{compressed} form as before, 
so that all $occ(p)$ occurrences of a 
pattern $p\in \Sigma^\ell$ in $x$ can be reported, or more modestly, counted, 
in near-optimal time $O(\ell + occ(p))$. The celebrated BWT-based 
compressed text index of Ferragina and Manzini  \cite{FM}, commonly known as the \emph{FM-index}, 
achieves the latter task using $|\HRLX(L)| + O(n/\lg n)$ 
bits of space\footnote{For \emph{reporting} queries, \cite{FM} requires significantly larger space 
$|\HRLX(L)|\cdot \lg^{\eps}n + O(n/\lg^{1 - \eps} n)$ for some $\epsilon \in (0,1)$.}, 
and is one of the most prominent tools in pattern matching for bioinformatics and text applications 
(see \cite{NavarroSurvey07} and references therein). 
The core of the FM-index is a compressed data structure that computes \textsc{Rank} queries over the (compressed) 
BWT string $L$ in constant time and 
redundancy $r=O(n/\lg n)$ bits, and more generally, in time $t'$ and redundancy $r = \Theta(n/(t'  \lg n))$. 
Their data structure, combined with simple ``marking index"\footnote{I.e., recording shortcut pointers to the 
location of $x_i$ in $L$ for every block $i \in [j\cdot t], \; j\in [n/t]$, see Section \ref{sec_technical_overview}.} 
of the BWT permutation on blocks of length $t$, 
yields a solution to  Problem 1  with overall redundancy $\tilde{O}\left(n/t + n/t'\right)$ and time $O(t \cdot t')$ 
(as simulating each sequential step of the BWT decoder requires $O(1)$ rank queries, 
and there are  $t$ coordinates per block). In other words, if the desired decoding time of a 
coordinate is $t$,  \cite{FM} gives redundancy $r=\tilde{O}(n/\sqrt{t})$ over $\HRLX$. 
In fact, even for \emph{randomized} succinct dictionaries approaching $H_k(x)$ bits of space, the best 
known trade-off is $r=\tilde{\Theta}(n/\sqrt{t})$: 
Dutta et al. \cite{Dutta+} gave a randomized data structure with space $(1+\eps)\tilde{H_k}(x)$\footnote{Here, $\tilde{H_k}$ 
denotes the Lempel-Ziv \cite{LZ78} codeword length $|\mathsf{LZ78}(x)|$, i.e., $\tilde{H_k}(x) \approx H_k(x) + \Theta(n/\lg n)$.}, 
and expected decoding time $\Theta(1/\eps^2)$ to retrieve each $x_i$. 
Our first main result is a near-quadratic improvement over this trade-off for Problem $1$: 

\begin{theorem}[Local Decoding of BWT, Informal]  \label{thm_local_bwt}    
For any $t$ and any string $x\in \Sigma^n$, 
 there is a succinct data structure  
that stores 
$ |\HRLX(\BWT(x))| + \tilde{O}\left(\frac{n \lg t}{t}\right) + n^{0.9}$ 
bits of space, so that each coordinate $x_i$ can be retrieved 
in $O(t)$ time, in the word-RAM model with word size $w=\Theta(\lg n)$. 
\end{theorem} 

Our data structure directly implies that a contiguous substring of size $\ell$ of $x$ can be decoded in time $O(t + \ell\cdot \lg t)$ 
without increasing the space redundancy. It is noteworthy that achieving a linear trade-off as above between time and 
redundancy with respect to \emph{zeroth order} entropy ($H_0(x)$) 
is trivial, by dividing $x$ into $n/t$ blocks of length $t$, and compressing each block using Huffman (or arithmetic) 
codes, as this solution would lose at most $1$ bit per block. 
This solution completely fails with 
respect to higher-order entropy benchmarks $H_k$ (in fact, even against $H_1$), since in the 
presence of contexts, the loss in compressing each block \emph{separately}  can be arbitrarily large  
(e.g., $H_0\left((ab)^{n/2}\right) = n$ but $H_1\left((ab)^{n/2}\right) = \lg n$).  
This example illustrates the qualitative difference between the \textsc{Dictionary} problem in the independent 
vs. correlated setting. In fact, we prove a complementary cell-probe lower bound of $r \geq \Omega(n/t^2)$
on Problem $1$ for ``symmetric" data structures, which  
further decode the dispositions of any $x_i$ in $L=\BWT(x)$ and vice versa, in time $O(t)$ 
(see Theorem \ref{thm_LB} below). 
While removing this (natural) restriction remains an interesting open question, 
this result provides a significant 
first step in understanding the cell-probe complexity of Problem $1$; more on this below.  

Our second main result, which is a by-product of our central data structure, is an \emph{exponential} 
improvement (in $t$) on the redundancy of the FM-index for compressed pattern-matching counting: 

\begin{theorem}
\label{thm_exponential_FM} 
There is a small constant $\delta > 0$ such that for any  
$x\in \Sigma^n$  and any $t \leq \delta \lg n$, 
there is a compressed index using 
$|\HRLX(\BWT(x))| + n \lg n/2^t + n^{1 - \Omega_\delta(1)}$ 
bits of space, counting the number of occurrences of any pattern $p\in \Sigma^{\ell}$
in time $O(t\ell)$.
\end{theorem} 

To the best of our knowledge, Theorem \ref{thm_exponential_FM} provides the first compressed text index 
for pattern-matching counting queries, that can provably go below the $\Omega(n/\lg n)$ space barrier 
while maintaining near-optimal query time. 
In particular, it implies that 
at the modest increase of query time by a $O(\lg\lg n)$ factor, 
\emph{any} $n/\poly\lg n$ redundancy is achievable. In the interesting 
setting of compressed pattern-matching,  
where the text compresses to $n^{1-o(1)}$ bits (e.g., $H_k(x)  
=n/\lg^{O(1)} n$), this 
result provides an exponential (in $t$) \emph{overall} (i.e., multiplicative)  space reduction over the FM-index. 
Compressed string-matching in the the $o(1)$ per-bit entropy regime was advocated in the 
seminal work of Farach and Thorup, see \cite{FT} and references therein. 
For \emph{reporting} queries, we obtain a 
quadratic improvement over the FM-index, 
similar to Theorem \ref{thm_local_bwt} (see Section \ref{sec_report_pattern_matching}). 

The main ingredient of both Theorem \ref{thm_local_bwt} and \ref{thm_exponential_FM} is a 
new succinct data structure for computing \textsc{Rank} queries over the compressed BWT string $L=\BWT(x)$, with 
\emph{exponentially} small redundancy $r \approx n/2^t$ with respect to $|\HRLX(L)|$ 
(see Theorem \ref{thm_exp_tradeoff_rk_L} below). 
Our data structure builds upon and is inspired by the work of \Pat  \cite{Pat}, who showed a similar
 exponential trade-off for the \textsc{Rank} problem, 
with respect to the \emph{zeroth order} entropy $H_0(L)$, i.e., in the \emph{context-free} setting. 
In that sense, our work can be viewed as a certain higher-order entropy analogue of \cite{Pat}. 

The most challenging part of our data structure is dealing with the \emph{Move-to-Front} (MTF) encoding of $L$.   
This adaptive coding method comes at a substantial price: 
decoding the $i$th character from its encoding $\MTF(x)_i$ requires the decoder to know the current 
``state" $S_i \in \cS_{|\Sigma|}$ of the encoder, namely, the precise \emph{order} of recently occurring symbols, 
which itself depends on the \emph{entire history} $x_{<i}$. 
This feature of the MTF transform, that the codebook itself is \emph{dynamic},  
is a qualitative difference from other compression schemes such as Huffman coding or (non-adaptive) arithmetic codes, 
in which the codebook is \emph{fixed}. In fact, in that sense the MTF transform is conceptually closer to  
the BWT transform itself than to Huffman or arithmetic codes, since in both transforms, decoding the $i$th character is a sequential function 
of the decoded 
values of previous characters. 
Fortunately, it turns out that MTF has a certain  local 
property that can be leveraged during preprocessing time, 
leading to the following key result: 

\begin{theorem}[Locally-decodable MTF, Informal] \label{thm_local_mtf} 
For any string $x\in \Sigma^n$, there is a succinct data structure that encodes $x$ using at most 
$H_0(\MTF(x)) + \frac{n}{\left(\lg(n)/t \right)^{t}} + n^{0.9}$
bits of space,  
such that $x_i$ can be decoded in time $O(t)$. 
Moreover, 
it supports 
\textsc{Rank} queries with the same parameters. 
\end{theorem} 
The additive $n^{0.9}$ term stems from the storage space of certain look-up tables, which are \emph{shared} across 
$n / r$ blocks of size $r = \left(\lg(n)/t \right)^{t}$. 
Hence these tables occupy $o(1)$ bits per block in an amortized sense. Theorem \ref{thm_local_mtf} therefore has the surprising corollary that,
with only $O(1)$ bits of redundancy per block (even with respect to the $\MTF$ code \emph{followed} by arithmetic coding), 
decoding time can be reduced from $\Omega(n)$ to $O(\lg n)$.

\smallskip 

\paragraph{Techniques.}
Our techniques are quite different from those used in previous compressed text indexing literature, 
and in particular from the work of \cite{FM}.  At a high-level, our main result uses a combination of succinct 
Predecessor search (under some appropriate encoding) along with 
ad-hoc design of two ``labeled" B-tree data structures, known as \emph{augmented aB-trees} (see Section \ref{sec_prelims_aBtrees} below). We then use a theorem of \Pat  \cite{Pat} to \emph{compress} these aB-trees, resulting in succinct representations while preserving the query time. 
Our main tree structure relies on a new local preprocessing step of the $\MTF$ codeword, 
which in turn facilitates a
\emph{binary-search} algorithm for the dynamic $\MTF$ codebook (stack state). We show that 
this algorithm can be `embedded' as a (compressed) augmented aB-tree.  
Our cell-probe lower bound (Theorem \ref{thm_LB}) relies on a new ``entropy polarization"  
lemma for BWT permutations, combined with a `nonuniform' adaptation of Golynski's cell-elimination 
technique for the succinct permutations problem \cite{Golynski}.



\subsection{Related Work} \label{subsec_related_work}
A large body of work  has been devoted to succinct data structures efficiently supporting \textsc{Rank/Select} 
and (easier) \textsc{Dictionary} queries under \emph{zeroth-order} empirical entropy,  i.e., using $H_0(x) + o(n)$ bits of space. 
In this ``\emph{context-free}" regime, early results showed a near-linear
$r=\tilde{O}(n/t)$ trade-off between time and redundancy 
(e.g., \cite{Pagh, Munro, GGGRR}), and this was shown by Miltersen \cite{Miltersen} to be optimal  
for \emph{systematic}\footnote{Systematic data structures are forced to store the \emph{raw} input database $x$,  followed by  
an $r$-bit additional \emph{index}.} 
data structures. 
This line of work culminated with a surprising result of \Pat \cite{Pat}, 
who showed that an \emph{exponential} trade-off between time and redundancy can be achieved 
using \emph{non-systematic} data structures in the word-RAM model, 
supporting all the aforementioned operations in query time $O(t)$ and $s \approx H_0(x) + O\left(n/(\lg(n)/t)^t\right)$ bits of space. 
For \textsc{Rank/Select}, this trade-off was shown to be optimal in the cell-probe model with word-size $w=\Theta(\lg n)$  \cite{PV09}, 
while for \textsc{Dictionary} (and \textsc{Membership}) queries, the problem is still open  \cite{Vio08,Vio09,DodisPatrascuThorup}. 
There are known string-matching data structures, based on context-free grammar compression 
(e.g., LZ or SLPs \cite{SLP}), that achieve logarithmic query time for \textsc{Dictionary} queries,  
at the price of linear (but \emph{not} succinct) space in the compressed codeword \cite{Dutta+,SLP}. 
However, these data structures have an $O(n/\lg n)$ additive space term, regardless of the compressed 
codeword, which becomes dominant in the $o(1)$ per-bit entropy regime, 
the interesting setting for this paper (see \cite{FT} for elaboration).

The problem of compressed pattern matching 
has been an active field of research  
for over four decades, since the works of McCreight  \cite{McCreight1976} and Manber and Myers 
 \cite{ManberMyers1993}, who introduced suffix trees and suffix arrays.
Ferragina and Manzini \cite{FM} were the first to achieve a \emph{compressed} 
text index (with respect to higher-order 
entropy $H_k$), supporting pattern-matching counting and reporting 
queries in \emph{sublinear} ($o(n)$) space and essentially optimal query time. 
Their BWT-based data structure, known as the \emph{FM-index},
is still widely used in both theory and practice, and its applications in genomics go well beyond 
the scope of this paper \cite{SD10}. 
Subsequent works, e.g. \cite{Grossi, FT, GRR},  
designed compressed text indices under other entropy-coding space benchmarks, such as Lempel-Ziv compression, 
but to the best of our knowledge, all of them again require $\Omega(n/\lg n)$ bits of space, 
even when the text itself compresses to $o(n/\lg n)$ bits.   
We remark that for \emph{systematic} data structures, linear trade-off ($r=\tilde{\Theta}(n/t)$) is the best possible 
for counting pattern-matches \cite{GM03,Golynski}, hence Theorem \ref{thm_exponential_FM} 
provides an exponential separation between systematic and non-systematic data structures  
for this problem. For a more complete state of 
affairs on compressed text indexing, we refer the reader to \cite{FM, NavarroSurvey07}.

Another related problem to our work is the \emph{succinct permutations} problem \cite{Viola17,Viola18, Golynski, MRRR},
where the goal is to succinctly represent a permutation $\pi \in S_n$ using $\lg n! + r$ bits of space, supporting 
evaluation $(\pi(i))$ and possibly inverse ($\pi^{-1}(i)$) queries in time $t$ and $q$ respectively. 
For the latter problem, an essentially tight trade-off $r = \Theta(n\lg n/tq)$ is known in the regime 
$t,q\in \tilde{\Theta}(\lg n)$ \cite{MRRR,Golynski}.  


\paragraph{Organization.}   
We start with some necessary background and preliminaries in Section \ref{sec_prelims}. 
Section \ref{sec_technical_overview} provides a high-level technical overview of our main results. 
Sections \ref{sec_MTF},\ref{sec_rle} contain our main data structure and proofs of Theorems \ref{thm_local_bwt},
\ref{thm_exponential_FM} and \ref{thm_local_mtf}.
Section \ref{sec_report_pattern_matching} describes application to improved pattern-matching (reporting) queries.  
In Section \ref{sec_LB} we prove the cell-probe lower bound (Theorem \ref{thm_LB}).


\section{Background and Preliminaries} \label{sec_prelims}

For an $n$-letter string $x \in \Sigma^n$, let $n_c$ be the number of occurrences, i.e., the \emph{frequency}, of the symbol $c\in \Sigma$ in $x$. 
For $1 \leq i < j \leq n$, let $x[i : j]$ denote the substring $(x_i, x_{i+1}, \cdots, x_j)$. For convenience, we use the shorthand $x_{<i}$ to denote 
the prefix $(x_1, x_2, \cdots, x_{i-1})$. The \emph{$k^{th}$ context} of a character $x_i$ in $x$ is the substring of length $k$ that precedes it. 
A \emph{run} in a string $x$ is a maximal substring of repetitions of the same symbol. For a compression algorithm $\cA$, we denote by 
$|\cA(x)|$ the output size \emph{in bits}. The \emph{zeroth order empirical entropy} of the string $x$ is $H_0(x) := \sum_{c\in \Sigma} n_c \lg \frac{n}{n_c}$ 
(all logarithms throughout the paper are base-$2$, where by standard convention, $0\lg 0 = 0$). It holds that $0\leq H_0(x) \leq n \lg |\Sigma|$.
For a substring $y\in \Sigma^k$, let $y_x$ denote the concatenated string consisting of the single characters following all occurrences 
of $y$ in $x$. The \emph{$k^{th}$ order empirical entropy} of $x$ is defined as $H_k(x) := \sum_{y\in \Sigma^k} H_0(y_x)$.
This prior-free measure intuitively captures ``conditional" entropies of characters in correlated strings with bounded context, and 
is a lower bound on the compression size $|\cA(x)|$ of any $k$-local compressor $\cA$, i.e., any compressor that encodes each symbol with a code that only depends on the symbol itself and on the $k$ immediately preceding symbols;
for elaboration see \cite{Manzini99}. For all $k \geq 0$, we have $H_{k+1}(x) \leq H_k(x)$. Note that the space benchmark 
$H_k$ can be significantly smaller than $H_0$. For example, 
for $x = (ab)^{n/2}$, $H_0(x) = n$ but $H_k(x) = 0$ for any $k\geq 1$
(assuming  the length $n$ is known in advance; $H_k(x) \leq \lg n$ otherwise). 
For a random variable $X\sim \mu$, $H(X)$ denotes the Shannon entropy of $X$. 
Throughout the paper, we assume the original alphabet size is $|\Sigma| = O(1)$.

\paragraph{Succinct data structures.}{   
We work in the word-RAM model of word-length $w=\Theta(\lg n)$, in which arithmetic and shift operations on memory words require 
$O(1)$ time. A \emph{succinct data structure} for an input $x \in \Sigma^n$ is a data structure that stores a small 
\emph{additive} space overhead $r=o(n)$ beyond the ``information-theoretic minimum'' space $h(x)$ required to represent $x$, 
while supporting queries efficiently.  
In the ``prior-free" setting, 
$h(x)$ is usually defined in terms of empirical entropy $H_k(x)$.   
The space overhead  $r$ is called the \emph{redundancy}, and is measured in \emph{bits}.}


\subsection{The Burrows-Wheeler Transform}

Given a string $x \in \Sigma^n$, the Burrows-Wheeler Transform of $x$, denoted $\BWT(x)$, is defined
by the following process.
We append a unique end-of-string symbol `\$', which is lexicographically smaller than any character in $\Sigma$, to $x$ to get $x\$$ (without this technicality, invertibility is only up to cyclic shifts).
We place all $n+1$ cyclic shifts of the string $x\$$ as the rows of an $(n+1) \times (n+1)$ matrix, denoted by $\hat{\cM}$.
Then we sort the rows of $\hat{\cM}$ in lexicographic order.
The sorted matrix, denoted $\cM$, is henceforth called the ``BWT matrix" of $x$.
Finally, we output 
$L \in (\Sigma \cup \{\$\})^{n+1}$, the last column of the BWT matrix $\cM$. We henceforth use the shorthand $L := \BWT(x)$.

We observe that every column in $\cM$ is a permutation of $x\$$. Let $F$ and $L$ be the \emph{first} and \emph{last column} of $\cM$ respectively. See an example in Figure \ref{fig_bwt} below. 
For ease of notation,
we shall refer to x\$ as $x$, denote its length by $n$, and include \$ in $\Sigma$.

\begin{figure}[h]
	\centering
	\begin{tabular}{l}
		\\
		\hline
		mississippi\textbf{\$} \\
		ississippi\textbf{\textbf{\$}}m \\
		ssissippi\textbf{\$}mi \\
		sissippi\textbf{\$}mis \\
		issippi\textbf{\$}miss \\
		ssippi\textbf{\$}missi \\
		sippi\textbf{\$}missis \\
		ippi\textbf{\$}mississ \\
		ppi\textbf{\$}mississi \\
		pi\textbf{\$}mississip \\
		i\textbf{\$}mississipp \\
		\textbf{\$}mississippi
	\end{tabular}
	\hspace{0.7cm} $\Longrightarrow$ \hspace{0.7cm}
	\begin{tabular}{c c c}
		F & & L \\
		\hline
		\textbf{\$} & mississipp & i \\
		i & \textbf{\$}mississip & p \\
		i & ppi\textbf{\$}missis & s \\
		i & ssippi\textbf{\$}mis & s \\
		i & ssissippi\textbf{\$} & m \\
		m & ississippi & \textbf{\$} \\
		p & i\textbf{\$}mississi & p \\
		p & pi\textbf{\$}mississ & i \\
		s & ippi\textbf{\$}missi & s \\
		s & issippi\textbf{\$}mi & s \\
		s & sippi\textbf{\$}miss & i \\
		s & sissippi\textbf{\$}m & i
	\end{tabular}
	\caption{Burrows-Wheeler Transform for the string $x = $``mississippi'', with the unsorted matrix $\hat{\cM}$ on the left and the sorted matrix $\cM$ on the right. The output is $L = \BWT(x) = $``ipssm\$pissii''.}
	\label{fig_bwt}
\end{figure}

\subsubsection{Decoding BWT and the ``LF Mapping"} \label{sec_decoding_BWT}
While not obvious at first glance, BWT is an \emph{invertible} transformation. 
An important first observation for this fact is that the \emph{first} column $F$ of the BWT matrix $\cM$  
is actually known ``for free" (as long as the frequencies of each symbol are stored, using negligible 
$O(|\Sigma|\lg n)$ additive space), since $\cM$ is sorted lexicographically (See Figure \ref{fig_bwt}).
To see why this is useful, we first introduce the following central definition: 

\begin{definition}[Rank of a character]\label{def_rank}
	Let  $y \in \Sigma^n$. For any $c \in \Sigma, i \in [n]$, $rk_y(c,i)$ denotes the number of occurrences of the symbol $c$ 
	in the $i^{th}$ prefix $y_{\leq i} = y[1 : i]$. 
\end{definition}

Note that $rk_L(c,n) = n_c$, recalling that $n_c$ is the frequency of $c$ in $x$. We define the \emph{Last-to-First (LF) column mapping} $\pi_{LF} : [n] \mapsto [n]$ by setting $\pi_{LF}(i) = j$ if the character $L_i$ is located at $F_j$, i.e., $L_i$ is the first character in the $j^{th}$ row of the BWT matrix $\cM$. 
We note that $\pi_{LF}$ is a permutation.

An indispensable feature of BWT, the \emph{LF Mapping Property}, states that for any character $c \in \Sigma$, the occurrences of $c$ in the first column $F$ and last column $L$ follow \emph{the same order}. In other words, the permutation $\pi_{LF}$ preserves the ordering among all occurrences of $c$.


\begin{fact}[LF Property] \label{lem_LF}  
	For $i \in [n], c \in \Sigma$, we have 
	$rk_L(c, i) = rk_F(c, \pi_{LF}(i)) = \pi_{LF}(i) - \sum_{c' < c} n_c$.
\end{fact}
The second equality follows directly from the fact that the first column $F$ is sorted lexicographically by construction, while the first equality also requires the fact that $L$ is sorted by its right context.
The formal argument can be found in Section \ref{sec_LF} of the Appendix. The LF Mapping Property leads to the following lemma, which is the heart of the BWT decoding algorithm:

\begin{lemma} \label{lem_lf_inverse}
	Fix a data structure $D$ that returns $rk_L(c, i)$ for given $i \in [n], c \in \Sigma$. Let $j \in [n]$. If we know the position $i$ of $x_j$ in $L$, then we can compute (even without knowing $j$) the character $x_j = L_i$, and (if $j \geq 2$) the position $i'$ of $x_{j-1}$ in $L$, with $O(|\Sigma|)$ calls to $D$.
\end{lemma}

\begin{proof}
	Given the position $i\in L$ of $x_j$, the character $L_i = x_j$ can be decoded via $2|\Sigma|$ rank queries on $L$, 
	by computing $rk_L(c, i) - rk_L(c,i-1) \; \forall \; c \in \Sigma$, which is nonzero only for $c^* := x_j$.  
	Now, given the rank $rk_L(c^*, i)$ of $x_j$ in $L$, the LF-property (Fact \ref{lem_LF}) allows us to translate it to the index 
	$i' := \pi_{LF}(i)$ of $x_j$ \emph{in F}. As such, $F_{i'} = x_j$. But this means that $L_{i'} = x_{j-1}$, as every row 
	of $\cM$ 
	is a cyclic shift of $x$ (i.e., in each row, $L$ and $F$ contain consecutive characters of $x$). 
\end{proof}

The decoding argument  asserts that a \textsc{Rank} data structure over $L$ allows us to ``move back'' \emph{one character} in 
$x$.
Note that 
the decoding algorithm implied by Lemma \ref{lem_lf_inverse} is inherently sequential: decoding a \emph{single} character $x_{n-i}$ 
of $x$ requires 
$O(|\Sigma|\cdot i)$ calls to $D$, hence $\Omega(n)$ worst-case time. 


\subsection{Compressing BWT} \label{sec_prelim_rlx}

\subsubsection{Move-to-Front encoding ($\MTF$)}
As mentioned in the introduction, when reasonably short contexts tend to predict a character in the input text $x$, the BWT string $L =\BWT(x)$ 
will exhibit local similarity, i.e, identical symbols will tend to recur at close vicinity. As such, we expect the integer string 
$\MTF(L)$ to 
contain many \emph{small integers}. This motivates  the following \emph{relative} encoding of $L$:

The \emph{Move-to-Front} transform (Bentley et al. 1986 \cite{Bentley}) replaces each character of $L$ with the \emph{number of distinct characters seen since its previous occurrence}. Formally, the encoder maintains a list, called the \emph{MTF-stack}, initialized with all characters $c\in \Sigma$ 
ordered alphabetically. To encode the $i^{th}$ character, the encoder outputs its \textsc{Rank} in the \emph{current stack} 
$S_{i-1} \in \cS_{|\Sigma|}$ (with the character at the top of $S_{i-1}$ having \textsc{Rank} $0$), and moves $c = L_i$ to the top of the stack, generating $S_i$. At any instant, the MTF-stack contains the 
characters ordered by recency of occurrence. Denote the output of this sequential algorithm by $m(L) := \MTF(L) = (m_1,m_2,\ldots, m_n) \in \{0, 1, \cdots, |\Sigma|-1\}^n$. 

A few remarks are in order: First, note that \emph{runs of identical characters} in $L$ are transformed into \emph{runs of $0$s} (except the 
first character in the run) in the 
resulting string $m(L)$. Second, at each position $i \in [n]$, the corresponding MTF-stack $S_i$ defines a unique \emph{permutation} 
$\pi_i \in \cS_{|\Sigma|}$ on $[|\Sigma|]$.

\subsection{The $\HRLX$ compression benchmark} \label{sec_RLX_comparison} 
Based on the $\MTF$ transform and following the original paper of Burrows and Wheeler \cite{BW}, 
\cite{FM,Manzini99} analyzed the following compression algorithm\footnote{Excluding the final arithmetic coding step.} 
to encode $L=\BWT(x)$, henceforth denoted $\HRLX(L)$: 
\begin{enumerate}
	\item Encode $L$ using the Move-to-Front transform to produce $\MTF(L)$.
	
	\item Denote by $L^{runs}$ the concatenation of substrings of $\MTF(L)$ corresponding to
	\emph{runs of $0$s}, 
	and by $\MTF(L^{-runs}) := [n] \setminus L^{runs}$ the remaining 
	coordinates. Encode all $0$-runs in $L^{runs}$ using 
	\emph{Run-Length encoding}, where each run is replaced
	by its length (encoded using a prefix-free code), and denote 
	the output by $\RLE\left(L^{runs}\right)$.
	
	\item Encode the remaining (non-runs) symbols in $\MTF(L^{-runs})$ using a $0$-order entropy code\footnote{A $0$-order encoder assigns a unique bit string to each symbol independent of its context, such that we can decode the concatenation of these bit strings.} (e.g., Huffman or Arithmetic coding), to obtain the 
	final bit stream $\HRLX(L)$ (suitably modified to be prefix-free over the alphabet comprising non-zero MTF symbols and run-length symbols).
\end{enumerate}	

See illustration in Figure \ref{fig_red_black_L}. For justification of the $\HRLX$ space benchmark and comparison to other compressors, we refer the reader to Section \ref{sec_app_RLX} of the Appendix.

\begin{figure}[H] 
	\centering
	\begin{tabular}{r p{0.23cm} p{0.23cm} p{0.23cm} p{0.23cm} p{0.23cm} p{0.23cm} p{0.23cm} p{0.23cm} p{0.23cm} p{0.23cm} p{0.23cm} p{0.23cm} p{0.23cm} p{0.23cm} p{0.23cm} p{0.23cm} p{0.23cm} p{0.23cm} p{0.23cm} p{0.23cm} p{0.23cm} p{0.23cm} p{0.23cm}}
		$L$ & $=$ & ``i & c & a & a & a & a & t & h & e & e & e & e & e & e & a & t & h & e & u & u & u & i".\\
		$\MTF(L)$ & = & (8,& 3,& 2,& \textbf{\blue{0}},& \textbf{\blue{0}},& \textbf{\blue{0}},& 19,& 9,& 7,& \textbf{\blue{0}},& \textbf{\blue{0}},& \textbf{\blue{0}},& \textbf{\blue{0}},& \textbf{\blue{0}},& 3,& 3,& 3,& 3,& 20,& \textbf{\blue{0}},& \textbf{\blue{0}},& 6).\\
		$\HRLX(L)$ & = & (8,& 3,& 2,& \textbf{\blue{3}},& 19,& 9,& 7, & \textbf{\blue{5}},& 3,& 3,& 3,& 3,& 20,& \textbf{\blue{2}},& 6.) & & & & & & 
	\end{tabular}
	\caption{Illustration of $\MTF$ and $\RLE$ encoding. $\RLE$ symbols in $\HRLX(L)$ are shown in bold blue. The final prefix-free code is not shown. $\MTF$ exploits local context in $L$, leading to small integers and in particular, $0$-runs. Each $0$-run is encoded by its length (a single character)
		in $\HRLX(L)$, yielding a significantly shorter string. Our \textsc{Rank} data structure is built over the compressed string $\HRLX(L)$, while it answers queries over the original string $L = \BWT(x)$.}
	\label{fig_red_black_L}
\end{figure}

By a slight abuse of notation, the output length of the algorithm\footnote{up to prefix-free coding overheads.} is 
\begin{equation}
	|\HRLX(L)| = |\RLE(L^{runs})| + \lceil H_0(\MTF(L^{-runs})) \rceil.
\end{equation}



\subsubsection{Augmented B-Trees  \cite{Pat}} \label{sec_prelims_aBtrees}

Central to our data structure is the notion of ``augmented B-trees", or \emph{aB-trees} for short. 
Let $B \geq 2$, $t \in \N$, and let $A \in \Sigma^s$ be an array of length $s := B^t$. An aB-tree $\cT$ over $A$ is a $B$-ary tree of depth $t$, with leaves corresponding to elements of $A$. Each node $v \in \cT$ is augmented with a value $\varphi_v$ from an alphabet $\Phi$. This value $\varphi_v$ must be a function of the subarray of $A$ corresponding to the leaves of the subtree $\cT_v$ rooted at $v$. In particular, the value of a leaf must be a function of its array element, and the value of an internal node must be a function of the values of its $B$ children.

The query algorithm starts at the root and traverses down the tree along a path which can be \emph{adaptive}. Whenever it visits a node, it reads all the values of its $B$ children and recurses to one of them, until it reaches a leaf node and returns the answer.	We ensure the query algorithm
spends $O(1)$ time per node, by packing all the augmented values of the children in a single word. 

For a given aB-tree $\cT$ and value $\varphi \in \Phi$, let $\cN(s, \varphi)$ be the number of possible arrays $A \in \Sigma^s$ such that the root is labeled with $\varphi$. A reasonable information-theoretic space benchmark for this data structure, conditioned on the root value 
$\varphi$, is therefore $\lg \cN(s, \varphi)$. 
\Pat \; proved the following remarkable result, which allows to \emph{compress} any aB-tree, while preserving its query time:
\begin{theorem}[Compressing aB-trees, \cite{Pat}]  \label{thm_pat_ab_compress}
	Let $B = O\left(\frac{w}{\lg
		(s + |\Phi|)}\right)$. We can store an aB-tree of size $s$ with root value $\varphi$ using $\lg
	\cN(s, \varphi) + 2$ bits. The query time is $O(\lg
	_B s) = O(t)$, assuming precomputed look-up tables of $O\left(|\Sigma| + |\Phi|^{B+1} + B \cdot |\Phi|^B \right)$ words, which only depend on $s, B$ and the aB-tree query algorithm.  
\end{theorem}
The proof idea is to use recursion in order to encode the root value $\varphi_r$, followed by 
an encoding of the augmented values $\varphi_v$ of every child of the root, \emph{``conditioned" on $\varphi_r$}, and so on, 
without losing (almost) any entropy (recursive encoding is needed to achieve this, since $\cN(s, \varphi_r)$ may not be a power of 2). 
Theorem \ref{thm_pat_ab_compress} allows us to represent any aB-tree with a redundancy of merely $2$ bits 
(over the \emph{zeroth-order} empirical entropy of the leaves).
Since the extra look-up tables do not depend on the array $A$, in our application, we use a similar trick as in \cite{Pat} and 
divide the original array of length $n$ into blocks of length $s = B^t$, building an aB-tree over each block. 
We then invoke Theorem \ref{thm_pat_ab_compress} 
separately on each tree, adding a certain auxiliary data structure that aggregates query answers across blocks so as to 
answer the query on the original array (for further details, see \cite{Pat}). 
Beyond facilitating the desired query time, this application renders the extra space occupied by the look-up tables in Theorem \ref{thm_pat_ab_compress} inconsequential, as they can 
be \emph{shared} across blocks. We remark that this ``splitting" trick of \cite{Pat} 
only applies when the augmented values $\varphi$ are \emph{composable}, in the sense that 
$\varphi(A\circ B) = f\left(\varphi(A),\varphi(B)\right)$, 
where $A\circ B$ is the concatenation of the arrays $A,B$. The aB-trees we design shall use 
augmented \emph{vector} values which are (component-wise) composable. 




\section{Technical Overview} \label{sec_technical_overview}

Both Theorem \ref{thm_local_bwt} and Theorem \ref{thm_exponential_FM} follow from the next result, which is 
the centerpiece of this work: 

\begin{restatable}
{thm}{rank} \label{thm_exp_tradeoff_rk_L}
There exists a small constant $\delta > 0$ such that for any  
$x\in \Sigma^n$ and $t' \leq \delta \lg n$, 
there is a succinct data structure $\cD_{rk}$ that supports \textsc{Rank} queries on $L = \BWT(x)$ in time $O(t')$, 
using at most 
$|\HRLX(L)|+ n \lg n / 2^{t'} + n^{1 - \Omega(1)}$
bits of space, 
in the $w=\Theta(\lg n)$ word-RAM model. 
\end{restatable}

Theorem \ref{thm_exponential_FM} is a direct corollary of Theorem \ref{thm_exp_tradeoff_rk_L}, 
as it turns out that counting the number of occurrences 
of a given pattern $p \in \Sigma^\ell$ in $x$,  
amounts to $O(\ell)$ successive \textsc{Rank} queries on $L$ (see \cite{FM}). 

To see how Theorem \ref{thm_local_bwt} follows from Theorem \ref{thm_exp_tradeoff_rk_L}, 
consider the following data structure for locally-decoding a coordinate $x_i$ of $x$ in time $t$: 
Let $t' < t$ be a parameter to be determined shortly. 
Let $\cD_{rk}$ be the data structure supporting rank queries on $L$ in time $O(t')$. We divide
$x$ into $\lceil n/T\rceil $ blocks 
of size $T := O(t/t')$, and store, for each ending index $j$ of a block, the position \emph{in $L$} corresponding to $x_j$. 
In other words, we simply record ``shortcuts" of the BWT transform after every block of size $T$.
Given an index $i \in [n]$,
the data structure first computes the endpoint $j := \left\lceil \frac{i}{T} \right\rceil T$ 
of the block to which $i$ belongs, reads from memory 
the position of $x_{j}$ in $L$, and 
then simulates $(j-i) \leq T = O(t/t')$ sequential steps of the LF-mapping decoding algorithm from 
Section \ref{sec_decoding_BWT}, to decode $x_i$. By Lemma \ref{lem_lf_inverse}, each step requires $O(|\Sigma|)$ \textsc{Rank} queries on $L$, each of which can be done using $\cD_{rk}$ in 
$O(t')$ time, hence the overall running time is $O(T\cdot t') = O(t)$.
To balance the redundancy terms, observe that 
the overall space of our data structure (up to $O(n^{\eps})$ terms) is 
\begin{align}\label{eq_loc_bwt_ds}
s = |\HRLX(L)|+ \frac{n \lg n}{2^{t'}} + \frac{n \lg n}{T}. 
\end{align}
Thus, setting $t' = \Theta(\lg t)$, leads to overall redundancy $r = O\left(\frac{n\lg n \lg t}{t} \right) = 
\tilde{O}\left(\frac{n\lg t}{t} \right)$, as claimed in Theorem \ref{thm_local_bwt}.
Next, we provide a high-level overview of the proof of Theorem \ref{thm_exp_tradeoff_rk_L}. 


\subsection{Proof Overview of Theorem \ref{thm_exp_tradeoff_rk_L}}

Recall (Section \ref{sec_prelim_rlx}), that the $\HRLX$ compression of $L=\BWT(x)$ can be 
summarized as : 
		\[  |\HRLX(L)| = |\RLE(L^{runs})| + \lceil H_0(\MTF(L^{-runs})) \rceil.\]
Since $\HRLX$ compresses 
the two parts $L^{runs}$ and $L^{-runs}$ using two conceptually different encodings ($\RLE$ and $\MTF$, respectively), 
it makes sense to design a \textsc{Rank} data structure for each part separately (along with an efficient mapping for combining  
the two answers to compute the overall rank of a character in $L$).  
This modular approach simplifies the presentation and, 
more importantly, enables us to achieve a significantly 
better redundancy for Theorem \ref{thm_local_mtf} (i.e., $n/(\lg n/t)^t$ instead of $n/2^t$), but is slightly suboptimal in terms of 
space (by an $\Omega(|\HRLX(L)|)$ additive term). In the actual proof, we show how the two data structures below 
can be ``merged" to avoid this overhead. 

\paragraph{A Rank data structure over $\RLE(L^{runs})$.}
Our first goal is to design a compressed data structure $\cD_{\RLE}$ that reports, for each symbol 
$c\in \Sigma$ and index $i \in [n]$, the number of occurrences of $c$ in $L[1:i]$ that are contained 
\emph{in $L^{runs}$}, i.e., the number of consecutive $0$'s in $MTF(L)[1:i]$ that correspond to runs 
of $c$. Since $\HRLX$ represents this substring by ``contracting" each run into a singleton 
(denoting its length), solving this problem succinctly essentially entails a Predecessor 
search\footnote{For a set of keys $S \subset \cU$ with $|S| = \kappa$, \textsc{Predecessor}$(i,S)$ 
returns $\max\{x \in S \text{ } | \text{ } x \leq i\}$.}
on the universe $[n]$ with $\kap=\kap(L)$ ``keys", where $\kap$ denotes the number 
of runs in $L$. Alas, under the standard representation of this input, as a $\kap$-sparse string in $\{0,1\}^n$, 
Predecessor search clearly requires at least $\lg{n \choose \kap}$ bits of space \cite{Pat,PT}, which could be 
$\gg |\HRLX(L)|$ (for example, when all but a single $0$-run are of constant length and separation, 
which is an oblivious feature to the previous representation). 
To adhere to the $\RLE$ space benchmark, we use 
a more suitable alternative representation of $L^{runs}$.

To this end, suppose for simplicity of exposition,   
that $L$ consists entirely of runs (i.e., $L= L^{runs}$), and  
that the character $c \in \Sigma$ corresponding to each 0-run is known at query time 
(this will be handled in the integrated data structure in Section \ref{sec_rle}). 
For $i \in [\kappa]$, let $\ell_i \in [n]$ denote the length of the $i^{th}$ run,  
and let $L' = (\ell_1,\ell_2,\ldots ,\ell_\kap) \in [n]^{\kap}$ be the string that encodes the run lengths. Note that 
$\HRLX$ spends precisely $\sum_i \lg \ell_i$ bits to encode this part (ignoring prefix-coding issues).
	
To compute $rk_{L^{runs}}(c,i)$, 
we design an adaptive augmented aB-tree, that essentially 
implements a predecessor search over the new representation $L'$ of $L^{runs}$: 
We first construct a $B$-tree $\cT$ 
over the array $L' \in [n]^{\kap}$, and augment each intermediate node $v$ of the tree with the (vector-valued) 
function $\varphi_{RLE}(v) := (\varphi^c_{\ell}(v))_{c \in \Sigma} \; \in [n]^{|\Sigma|}$
, where $\varphi^c_{\ell}(v)$ counts the total 
sum $\sum_{j \in \cT_v} \ell_j$ of run-lengths in the subtree of $v$, corresponding to runs of $c$. Given an index $i\in [n]$ and character $c \in \Sigma$, 
the query algorithm iteratively examines the labels of all $B$ children of a node $v\in \cT$ starting from the root, and recurses 
to the rightmost child $u$ of $v$ for which $\sum_c\varphi^c_\ell(u) \leq i$ (i.e., to the subtree that contains 
the interval $\ell_j$ to which $i$ belongs), collecting the sum of $\varphi^c_\ell (u)$'s along the query path.

To ensure query time $O(t')$, we break up the array as in \cite{Pat} into sub-arrays each of size $B^{t'}$ (for $B=\Theta(1)$),
and build the aforementioned tree over each sub-array (this is possible since the augmented vector $\varphi_{RLE}$ 
is a (component-wise) composable function). To ensure the desired space bound for representing $\cT$, we further augment each node $v$ with a 
``zeroth-order entropy" constraint $\varphi_0(v)$, counting the sum of marginal 
empirical entropies  
$n^v_c \lg (n_c/n)$\footnote{For $c \in \Sigma$ and node $v$, $n^v_c$ denotes the frequency of $c$ in the sub-array rooted at $v$.} of the elements in the subtree $\cT_v$ (which can be done recursively due to 
additivity of $\varphi_0$ w.r.t $v$'s). 
A standard packing argument then ensures $\cN(\kappa, \varphi) \leq 2^{\varphi_0(r)} \lesssim 2^{H_0(\cT_r)}$ (where $\cT_r$ is the sub-array rooted at $r$), 
as desired. We then invoke Theorem \ref{thm_pat_ab_compress} to compress $\cT$ to $H_0(L') + O\left(\frac{n \lg n}{B^{t'}}\right)$ bits, 
yielding exponentially small redundancy (up to $n^{1-\eps}$ additive terms). 
This ensures that 
the total space (in bits) occupied by $\cD_{\RLE}$ is essentially  
\[ H_0(L') + O \left(\frac{n \lg n}{B^{t'}}\right) \leq |\RLE(L^{runs})| + O\left(\frac{n \lg n}{B^{t'}}\right). \] 

The actual proof is slightly more involved, since the merged data structure needs to 
handle characters from both $L^{runs}$ and $MTF(L^{-runs})$ simultaneously, hence it must efficiently 
distinguish between 0-runs corresponding to different symbols. Another issue is that  
Theorem \ref{thm_pat_ab_compress} of \cite{Pat} is only useful for truly sub-linear alphabet sizes, whereas  
$(L')_i\in [n]$, hence in the actual proof we must also split long runs into chunks of length $\leq n^\eps$. 
A simple application of the log-sum inequality ensures 
this truncation does not increase space by 
more than an $\tilde{O}(n^{1-\eps})$ additive term.

\paragraph{A Rank data structure over $\MTF(L^{-runs})$.}
The more challenging task is computing $rk_{L^{-runs}}(c,i)$, 
i.e., the frequency of $c$ in $L[1:i]$ contained 
in the substring $\MTF(L^{-runs})$, which is obtained by applying the MTF transform to $L$ and deleting all $0$-runs 
(see Figure \ref{fig_red_black_L}). 
Note that the mapping from $i \in L$ to its corresponding index $i' \in \MTF(L^{-runs})$ amounts to 
subtracting all runs before $i$. This operation can be performed using a single 
partial-sum query to our integrated data structure (in Section \ref{sec_rle}),    
which collects the sum of $\varphi^c_\ell(u)$'s \emph{over all} $c \in \Sigma$ along the query path.

As discussed in the introduction, the adaptive nature of the MTF encoding has the major drawback 
that decoding the $j^{th}$ symbol $\MTF(L^{-runs})_j$,  
let alone computing its rank, requires knowing the corresponding MTF stack state $S_{j-1} \in \cS_{|\Sigma|}$ (i.e., the 
precise order of recently occurring symbols), which itself depends on the \emph{entire} history $L^{-runs}_{<j}$. 
A straightforward solution is to store the MTF stack 
state after every block of length $t'$ (where $t'$ is the desired query time), much like the ``marking" solution for decoding BWT, yielding a 
linear search for the stack-state $S_j$ from the nearest block, and thus a linear time-space trade-off. 

To speed up the search for the local stack-state, we observe the following key property of the MTF transform: 
Let $\MTF(x) := (m_1, m_2,\ldots , m_n)$ be the MTF transform of $x\in \Sigma^n$ (see Figure \ref{fig_red_black_L} for illustration).     
Let $\cI=[i,j]$ be any sub-interval of $[n]$, and denote by $S_{i-1},S_j \in \cS_{|\Sigma|}$ the corresponding 
stack-states at the start and endpoints of $\cI$. Now, consider the \emph{permutation} $\pi_\cI := 
\mathbf{Id_{\Sigma}} \mapsto \hat{S_j}$, 
obtained by simulating the MTF decoder on $(m_i,\ldots, m_j)$ \emph{starting from the identity} state $\mathbf{Id_{\Sigma}}$, i.e., ``restarting" the MTF decoding algorithm but running it on the \emph{encoded} 
substring $(\MTF(x)_i,\ldots, \MTF(x)_j)$, arriving at some final state ($\hat{S_j}$) at the end of $\cI$   
(note that this process is well-defined). 
Then the true stack-state $S_j$ satisfies:  $S_j = \pi_\cI \circ S_{i-1}$. The crucial point is that 
$\pi_\cI$ is \emph{independent} of the (true) stack state $S_{i-1}$, i.e., it is a \emph{local} function of $\MTF(x)_\cI$ only. 
 
We show that this ``decomposition" 
property of the MTF transform (Proposition \ref{perm independent of stack}), facilitates a   
 \emph{binary search} for the local stack-state $S_{j-1}$ (rather than linear-searching) with very little space overhead, as follows: 
 At preprocessing time, we build an augmented $B$-tree 
 over the array $\MTF(x_1,\ldots,x_n)$, where each intermediate node $v$ is augmented with the \emph{permutation} $\pi_v \in \cS_{|\Sigma|}$ 
 corresponding to its subtree $\MTF(\cI_v)$, obtained by 
 ``\emph{restarting}" the MTF decoder to the identity state $\mathbf{Id_{\Sigma}}$, and  
 simulating the MTF decoder from start to end of $\cI_v$, 
 as described above. 
 Note that this procedure is well defined,  
 and that the aforementioned observation is crucially used here, as the definition of aB-trees requires each augmented 
 value of an intermediate node to be a \emph{local function} of its own subtree. 
 At query time, the query algorithm traverses the root-to-leaf$(j)$ path,  
 \emph{composing} the corresponding (possibly inverse) permutations between the stack-states along the path, 
 depending on whether it recurses to a right or left subtree. We show this process ensures that when the query algorithm reaches 
 the leaf $\MTF(x)_j$, 
it possesses the correct stack-state $S_{j-1}$, and hence can correctly decode $x_j$. 
 While this algorithm supports only ``local decoding" (\textsc{Dictionary}) queries, with an extra simple trick, 
 the above property in fact facilitates a similar aB-tree supporting 
\textsc{Rank} queries under the MTF encoding (see Section \ref{sec_MTF}). 


Once again, in order to impose the desired space bound ($\approx H_0(\MTF(L^{-runs}))$) and to enable 
arbitrary query time $t'$, 
we augment the nodes of the tree with an additional 
zeroth-order entropy constraint, and break up the array into sub-arrays of size $\Theta(B^{t'})$, this time for 
$B \approx \frac{\lg n}{t'}$. 
Compressing each tree using Theorem \ref{thm_pat_ab_compress}, and adding an auxiliary data structure to 
aggregate query answers across sub-arrays, completes this part and establishes Theorem \ref{thm_local_mtf}.

\subsection{Lower Bound Overview} 

We prove the following cell-probe lower bound for a somewhat stronger version of Problem \ref{problem_1}, 
which requires the data structure to 
efficiently decode \emph{both} 
forward and \emph{inverse} dispositions of the induced BWT permutation between 
$X$ and $L := \BWT(X)$ (we note that both the FM-Index and our data structure from 
Theorem \ref{thm_local_bwt} satisfy this natural 
requirement\footnote{I.e., for these data structures, 
we can achieve $q=O(t)$ by increasing the redundancy $r$ by a mere factor of $2$.}, 
and elaborate on it in Section \ref{sec_LB}): 

\begin{restatable}
	{thm}{thmLB} \label{thm_LB}
Let $X\in_R \{0,1\}^n$ and let $\Pi_X \in \cS_n$ be the induced BWT permutation from indices in 
$L := \BWT(X)$ to indices in $X$. 
Then, any data structure that computes  $\Pi_X(i)$ and $\Pi^{-1}_X(j)$ for every 
$i,j\in [n]$ in time $t, q$ respectively, such that $t \cdot q \leq \delta \lg n/\lg \lg n$ (for some constant $\delta > 0$), 
in the cell-probe model with word size $w=\Theta(\lg n)$, 
must use $n + \Omega\left(n/tq\right)$ bits of space in expectation.
\end{restatable}

We stress that Theorem \ref{thm_LB} is more general, as our proof can yield nontrivial  
lower bounds against general (non-product) distributions $\mu$ on $\Sigma^n$ with ``sufficient block-wise independence",  
though a lower bound against uniform strings is in some sense stronger, as it states that the above redundancy 
cannot be avoided even if $\Pi_X$ is stored in \emph{uncompressed} form 
(see also Section \ref{sec_LB}). 

Our proof of Theorem \ref{thm_LB} is based on a ``nonuniform" variation of 
the ``cell-elimination" technique of \cite{Golynski}, who used it to prove a lower bound of $r \geq \Omega(n\lg n/tq)$ 
on the space redundancy of any data structure for the \emph{succinct permutations} problem 
$\textsc{Perms}_n$. In this problem, the goal is to represent a random permutation $\Pi \in_R \cS_n$ 
succinctly using $ \lg n! + o(n \lg n)$ bits of space, supporting forward and inverse evaluation 
queries in query times $t,q$ respectively, as above. 
Alas, \cite{Golynski}'s compression argument crucially requires that 
\begin{align}\label{cond_gol}
t,q \leq O\left(\frac{H(\Pi)}{n\cdot \lg \lg n}\right).  
\end{align}
When $\Pi$ is a \emph{uniformly random} permutation, i.e., $H(\Pi) \approx n\lg n$, 
this condition implies that the lower bound holds 
for $t,q \leq O(\lg n/\lg\lg n)$. In contrast, the BWT permutation of $X$ can have \emph{at most 
$n\lg |\Sigma| = O(n)$} bits of entropy 
for constant-size alphabets (as $\Pi_X$ is determined by $X$ itself),  
hence condition \eqref{cond_gol} does not yield \emph{any} lower bound whatsoever for our problem.

To overcome 
this
obstacle, we prove an ``entropy polarization" lemma for BWT: 
It turns out that for a random string $X$, while an \emph{average} coordinate $\Pi_X(i)$   
indeed carries only $O(1)$ bits of entropy, the entropy distribution has huge variance.
In fact, we show that for any $\eps \geq \tilde{\Omega}(1/\lg n)$, 
there is a subset $\cI$ of only $(1 - \eps)\frac{n}{\lg n}$ coordinates in $[n]$,  
whose total entropy is $H(\Pi_{X}(\cI)) \geq (1-O(\eps))n$, 
i.e., this small set of coordinates has maximal entropy ($\approx \lg n$ bits each), and 
essentially determines the entire BWT permutation\footnote{Note that here we view $\Pi_X$ as a mapping from $X$ to $L = \BWT(X)$ 
and not the other way around, but this is just for the sake of simplicity of exposition and looking at $\Pi^{-1}_X$ is of course equivalent.}. 
This lemma (Lemma \ref{lem_ent_var}) is reminiscent of 
\emph{wringing lemmas} in information theory \cite{Wringing}, 
and may be a BWT property of independent interest in other applications.

The intuition behind the proof is simple: Consider dividing $X$ into 
$s:=\frac{n}{C\lg n}$ disjoint blocks of size $C\lg n$ each, and let $\cI := \{I_i,\ldots , I_s\} \subset [n]$
denote the set of 
first coordinates in each block respectively.
Since $X$ is random, each of the $s$ blocks is an \emph{independent}  
random ($C\lg n$)-bit string, hence for a large enough constant $C$, with overwhelming probability these substrings will be distinct,  
and in particular, their lexicographic order will be uniquely determined. 
Conditioned on this likely event, this lexicographic ordering remains random, hence the BWT  
locations of these indices alone must recover this random ordering, which is worth $\Omega(s \lg s) = \Omega(n)$ 
bits of information. However, the birthday paradox requires that $C>2$ to avoid collisions, in which case 
the above argument can only show that a small constant fraction ($< 0.5 n$) of the total entropy 
can be ``extracted" from this small set, 
while the remaining $n - o(n)$ coordinates possess most of 
the entropy. 
Unfortunately, this guarantee is too weak for our ``cell-elimination" argument, and would only yield 
a trivial $\Omega(n)$ space lower bound,  while we are seeking a lower bound on the additive redundancy beyond $n$. 

To bypass this, we observe that setting $C = (1+\eps)$, the number of ``colliding" blocks (i.e., non-distinct substrings 
of length $(1+\eps) \lg n$) is still only $\tilde{O}(n^{1-\eps}) \ll \eps n/\lg n$ with high probability. 
Moreover, we show that conditioned on 
this event $\cE$, the lexicographic ordering among the remaining \emph{distinct} $\approx (1-2\eps) \frac{n}{\lg n}$
blocks remains random. 
(For uniform $n$-bit strings,  conditioning on $\cE$ preserves exact 
uniformity of the ordering, by symmetry of $\cE$ w.r.t block-permutation, but more generally, 
we note that for any prior distribution, conditioning on $\cE$ does not ``distort" the 
original distribution by more than $\approx \sqrt{\lg(1/\Pr[\cE])}=o(1)$ in statistical distance,  
hence this argument can be generalized to nonuniform strings).  
Since, conditioned on $\cE$, 
the BWT mapping on $\cI$ 
determines the lexicographic ordering of the blocks, the data processing inequality (DPI) 
implies that the entropy of $\Pi_X(\cI)$ is at least 
$\approx (1-2\eps)\frac{n}{\lg n} \cdot \lg \frac{n}{\lg n} \geq (1-3\eps)n$, as claimed.   

Applying the ``entropy polarization" lemma with $\eps = O(\lg \lg n/\lg n)$, we then show how to adapt 
\cite{Golynski}'s cell-elimination argument to nonuniform permutations, 
deleting `unpopular' cells and replacing them with an efficient encoding of the \emph{partial} bijection $\Pi_X(\cI)$ 
between (forward and inverse) queries $\in \cI$ probing these cells. The polarization lemma then 
ensures that the remaining map of $\Pi_X$ on $\bar{\cI} = [n] \setminus\cI$ can be encoded directly using 
$H(\Pi_X(\bar{\cI}) | \;\cI, \Pi_X(\cI)) \leq O(\eps n) = O(n\lg\lg n/\lg n)$ bits, which will be dominated by the redundancy 
we obtain from the compression argument (as long as $tq \lesssim \lg n/\lg\lg n$),  thereby completing the proof.

\section{A Locally Decodable MTF Code and Rank Data Structure} 
\label{sec_MTF}

In this section, we prove the following theorem, which is a more formal version of Theorem \ref{thm_local_mtf}.

\begin{theorem} \label{thm_local_mtf_formal}
For any string $x\in \Sigma^n$ with $|\Sigma| = O(1)$, there is a succinct data structure that encodes $x$ using at most 
\[
H_0(\MTF(x)) +  n \Big{/} \left(\frac{\lg n}{\max(t, \lg \lg n)}\right)^t + n^{1 - \Omega(1)}
\]
bits of space, 
supporting \textsc{Rank} and \textsc{Dictionary} queries in time $O(t)$, in the word-RAM model with word size $w=\Theta(\lg n)$. 
\end{theorem}	
	
\paragraph{Setup and Notation.}
Let the alphabet $\Sigma = \{c_1, c_2, \cdots, c_{|\Sigma|}\}$, where $c_1 < c_2 < \cdots < c_{|\Sigma|}$ according to the lexicographical ordering on $\Sigma$. Let $S = (a_1, a_2, \cdots, a_{|\Sigma|})$ denote the MTF stack with $a_1$ at the top and $a_{|\Sigma|}$ at the bottom. For $j \in [|\Sigma|]$, let $S[j]$ denote the character at position $j$ in $S$, starting from the top. Fix a string $x = (x_1, x_2, \cdots, x_n) \in \Sigma^n$. Let $m = \MTF(x) = (m_1, m_2, \cdots, m_n) \in \{0, 1, \cdots, |\Sigma|-1\}^n$ be the Move-to-Front (MTF) encoding of $x$, with the initial MTF stack $S_0 := \left(c_1, c_2, \cdots, c_{|\Sigma|}\right)$. 

Given a MTF stack $S = (a_1, a_2, \cdots, a_{|\Sigma|})$ and a permutation $\pi \in \cS_{|\Sigma|}$, let $S' = \pi \circ S$ be the stack such that $S'[\pi(j)] = S[j] = a_j$ for all $j \in [|\Sigma|]$. 
We also associate with $S$ the permutation $\pi(S)$ which converts the initial stack $S_0$ to $S$, i.e., $S = \pi(S) \circ S_0$.
In this sense, we say that $S_0$ corresponds to the identity permutation $\mathbf{Id_{|\Sigma|}}$ on $[|\Sigma|]$, as
$S_0[j] = c_j$ for all $j \in [|\Sigma|]$. 
For $i \in [n]$, let $S_i$ be the stack induced by simulating the MTF decoder on $m[1 : i]$, starting from $S_0$. Equivalently, $S_i$ is the stack induced by 
$\MTF(x[1:i])$, i.e., the stack just after encoding the first $i$ characters of $x$, starting from $S_0$. For $0 \leq i < j \leq n$, let $\pi_{i, j} \in \cS_{|\Sigma|}$ be the unique permutation induced by simulating the MTF decoder on $m[i+1:j]$, starting from $S_i$.

\subsection{Properties of MTF Encoding}

The following proposition shows that for any $0 \leq i < j \leq n$, the permutation $\pi_{i, j}$ is a \emph{local} function of $m[i+1:j]$. So, these permutations $\pi_{i,j}$ are valid augmented values for an aB-tree built over $m = \MTF(x)$, without reference to the true MTF stacks $S_i$ and $S_j$.

\begin{proposition} \label{perm independent of stack}
	Fix $0 \leq i < j \leq n$, and let $S_i, S_j$ and $\pi_{i, j} \in \cS_{|\Sigma|}$ be as defined above. Then $\pi_{i, j}$ is independent of $S_i$ and $S_j$, given $m[i+1:j]$. Hence, we can generate $\pi_{i,j}$ by simulating the MTF decoding algorithm on $m[i+1:j]$, starting from the identity stack $S_0$.
\end{proposition}

\begin{proof}
	We prove this proposition by induction on $j-i$. Consider the base case, when $j - i = 1$. Then by definition of a single MTF step, we have
    \begin{equation}
    	\pi_{i,i+1}(k) = \begin{cases}
        k + 1 &\text{ if } k \leq m_{i+1}\\
    	1 &\text{ if } k = m_{i+1} + 1\\
        k &\text{ if } k > m_{i+1} + 1
    	\end{cases}
    \end{equation}
    Clearly, $\pi_{i, i + 1}$ is independent of $S_i$ and $S_{i+1}$ given $m_{i+1}$. This proves the base case.
    
    Now, suppose the claim is true for all $i, j$ such that $j - i = k \in \N$, and let $i, j$ be such that $j - i = k + 1$. Then by the induction hypothesis, $\pi_{i, j - 1}$ is independent of $S_i$ and $S_{j - 1}$ given $m[i+1:j-1]$. Moreover, $\pi_{j-1, j}$ is independent of $S_{j-1}$ and $S_j$ given $m_j$. Due to the sequential nature of the MTF encoding, we clearly have
    \[
    	\pi_{i, j} = \pi_{j-1, j} \circ \pi_{i, j-1}
    \]
    As both the permutations $\pi_{j-1, j}$ and $\pi_{i, j-1}$ are independent of stacks $S_i, S_{j-1}, S_j$ given $m[i+1:j]$, the same must be true for $\pi_{i,j}$.
\end{proof}

The following expression captures the evolution of the MTF stack, for all $0 \leq i < j \leq n$:
\begin{equation} \label{eqn_stack_evolution_forward}
S_j = \pi_{i, j} \circ S_i
\end{equation}

We can also ``reverse'' the steps of the MTF encoding. For fixed $0 \leq i < j \leq n$, if we are given the final stack $S_j$ and the permutation $\pi_{i, j}$, we can recover the initial stack $S_i$ by inverting $\pi_{i, j}$:
\[
	S_i = \pi^{-1}_{i, j} \circ S_j \label{stack evolution equation backward}
\]

]

\subsection{Locally Decodable MTF Code} \label{sub_sec_local_decoding_MTF}
We first describe the construction of a single aB-tree $\cT$ over the entire MTF encoding $m = \MTF(x) \in \{0, 1, \cdots, |\Sigma|-1\}^n$, which supports ``local decoding'' (\textsc{Dictionary}) queries.
Let $B \geq 2$ be the branching factor. Each node $v$ will be augmented with a permutation $\varphi_\pi(v) \in \Phi_\pi = \cS_{|\Sigma|}$. For $i \in [n]$, the leaf node $v$ corresponding to $m_i$ is augmented with the permutation $\varphi_\pi(v) = \pi_{i-1, i}$. Let $v$ be an internal node with its children being $v_1, v_2, \cdots, v_B$ in order from left to right. Then $v$ is augmented with the composition of permutations of its children, i.e.,
\begin{equation} \label{eqn_defn_pi}
	\varphi_\pi(v) = \varphi_\pi(v_B) \circ \varphi_\pi(v_{B-1}) \circ \cdots \circ \varphi_\pi(v_1).
\end{equation}
It is easy to observe that a node $v$ whose subtree $\cT_v$ is built over the sub-array $m[i+1:j]$ is augmented with the value $\varphi_\pi(v) = \pi_{i, j}$. Now, Proposition \ref{perm independent of stack} ensures that this is a legitimate definition of an aB-tree, because the value of a leaf is a function of its array element, and the value of an internal node is a function of the values of its $B$ children.

The query algorithm maintains a MTF stack $S$, which is initialized to the identity stack $S_0$ at the beginning of the array. Let $i \in [n]$ be the query index. The algorithm traverses down the tree, updating $S$ at each level. It maintains the invariant that whenever it visits a node $v$ whose sub-tree encompasses $m[j+1:k]$, it updates $S$ to the true stack $S_j$ just before the beginning of $m[j+1:k]$.

We describe how to maintain this invariant recursively. The base case is the root (at depth $d = 0$) whose subtree contains the entire array $m$. So, the query algorithm initializes $S = S_0$, which corresponds to the true initial MTF stack. Now, let $v$ be a node at depth $d$ whose sub-tree $T_v$ encompasses $m[j+1:k]$. Suppose the query algorithm has visited $v$, and $S$ is the true MTF stack $S_j$.  By assumption, $j+1 \leq i \leq k$. Let $v_1, v_2, \cdots, v_B$ be the children of $v$ in order from left to right, and let $v_{\beta^*}$ be the child of $v$ whose sub-tree includes $i$. Then we update $S$ as follows:
\begin{equation} \label{eqn_mtf_stack_evolution}
	S \leftarrow \varphi_\pi(v_{\beta^* -1}) \circ \varphi_\pi(v_{\beta^*-2}) \circ \cdots \circ \varphi_\pi(v_1) \circ S.
\end{equation}
 The above procedure explains the update rule which maintains the invariant at a node at depth $d+1$, assuming the invariant was maintained at a node at depth $d$. Thus, the proof that the invariant is maintained follows by induction on $d$.
 
Eventually, the algorithm reaches the leaf node corresponding to $m_i$. At this point, the MTF stack $S$ is the true stack $S_{i-1}$. Hence, it reports $x_i = S[m_i]$. The running time is $t = O(\lg_B n)$.

For the sake of simplicity, we have stated the update rule \ref{eqn_mtf_stack_evolution} purely in terms of forward compositions of permutations $\pi_{i,j}$. In practice, if $\beta^* > B/2$, one can equivalently update $S$ by starting from $\varphi_v \circ S$ and composing the inverse permutations $\varphi_\pi^{-1}(v_\beta)$ for $\beta \geq  \beta^*$:
\[
	S \leftarrow \varphi^{-1}_\pi(v_{\beta^*}) \circ \varphi_\pi^{-1}(v_{\beta^*+1}) \circ \cdots \circ \varphi_\pi^{-1}(v_{\beta}) \circ \varphi_\pi(v) \circ S.
\]
 However, since all permutations $\varphi_\pi(v_\beta), \beta \in [B]$ are stored in a word, both update rules take $O(1)$ time, and so the query time remains unaltered. Henceforth, we will continue to state the update rules purely in terms of forward compositions.

\subsection{Extension 
to Rank Queries (over MTF)} \label{sub_sec_rank_MTF}
The aB-tree $\cT$ above only supports ``local decoding" (\textsc{Dictionary}) queries over $m = \MTF(x)$, while our application requires answering 
\textsc{Rank} queries. 
We now show how $\cT$ can indeed be extended, via an additional simple observation, 
to support \textsc{Rank} queries under the MTF encoding.

Let $v$ be a node in the aB-tree $\cT$, whose subtree $\cT_v$ is built over the sub-array $m[i+1:j]$. We would like to augment $v$ with a vector $\tilde{\varphi}_{rk}(v) = \left(\tilde{\varphi}_{rk}(v, c_\sigma)\right)_{\sigma \in [|\Sigma|]} \in \{0, 1, \cdots, n\}^{|\Sigma|}$, such that $\tilde{\varphi}_{rk}(v, c_\sigma)$ is the frequency of the character $c_\sigma \in \Sigma$ in $x[i+1:j]$. However, as $\cT$ is built over $m = \MTF(x)$, and two occurrences of the same character $c \in \Sigma$ can be assigned distinct symbols in the MTF encoding, these augmented values are not consistent with the definition of an aB-tree.

To resolve this difficulty, we again use the fact that the permutation $\pi_{i, j}$ depends only on the sub-array $m[i+1:j]$. Recall that $S_0 = \left(c_1, c_2, \cdots, c_{|\Sigma|}\right)$ corresponds to the identity permutation $\mathbf{Id_{|\Sigma|}}$. For a node $v$, let $\varphi_{rk}(v) := \left(\varphi_{rk}(v, \sigma)\right)_{\sigma \in [|\Sigma|]}$, where $\varphi_{rk}(v, \sigma)$ is the frequency of $c_\sigma$ in the sub-array rooted at $v$, \emph{assuming} the MTF stack at the beginning of this sub-array is $S_0$. For a leaf node $v$ at $i \in [n]$, we have $\varphi_{rk}(v, \sigma) = 1$ if $m_i = \sigma - 1$, and $0$ otherwise.

Now, let $v$ be an internal node with children $v_1, v_2, \cdots, v_B$. Fix a character $c_\sigma \in \Sigma$. In general, the $\MTF$ stack at the beginning of the sub-array rooted at $\cT_v$ will be different from the $\MTF$ stack at the beginning of the sub-array $\cT_{v_\beta}$ rooted at each child $v_\beta, \beta > 1$. So, in order to express $\varphi_{rk}(v, \sigma)$ in terms of the values of its children, we need to 
add the entry of the vector $\varphi_{rk}(v_\beta)$ which corresponds to $c_\sigma$, for each $\beta \in [B]$.
We do this using the permutations $\varphi_{\pi}(v_\beta), \beta \in [B]$. For $\beta \in [B]$, the true MTF stack at the beginning of the sub-array rooted at $v_\beta$, assuming the MTF stack at the beginning of the sub-array rooted at $v$ is $S_0$, is given by Equation \ref{eqn_mtf_stack_evolution}. So, we have 
\begin{equation} \label{eqn_defn_rank}
\varphi_{rk}(v, \sigma) = \sum_{\beta=1}^B \varphi_{rk}(v_\beta, \varphi_{\pi}(v_{\beta-1}) \circ \varphi_{\pi}(v_{\beta-2}) \circ \cdots \circ \varphi_{\pi}(v_1)(\sigma))
\end{equation}

Let $\Phi_{rk} = \{0, 1, \cdots, n\}^{|\Sigma|}$. We augment each node $v$ with $\varphi_{rk}(v) \in \Phi_{rk}$. As we also encode the permutation $\varphi_\pi(v)$, the value at each internal node is a function of the values of its children, and hence this is a legitimate aB-tree.

The query algorithm, given $(c_\sigma, i) \in \Sigma \times [n]$, initializes a rank counter $rk = 0$, and traverses the same root-to-leaf path as before. Fix  an internal node $v$, with children $v_1, v_2, \cdots, v_B$, in its path. Let $\beta^* \in [B]$ be such that the sub-array rooted at $v_{\beta^*}$ contains the index $i$. The algorithm updates $rk$ as follows:
\begin{equation} \label{eqn_rank_update_rule}
rk \leftarrow rk + \sum_{\beta=1}^{\beta^*-1} \varphi_{rk}(v_\beta, \varphi_{\pi}(v_{\beta-1}) \circ \varphi_{\pi}(v_{\beta-2}) \circ \cdots \circ \varphi_{\pi}(v_1)(\sigma))
\end{equation}
Then it recurses to $v_{\beta^*}$ and performs this step until it reaches the leaf and returns $rk_x(c_\sigma, i)$.

\subsection{Compressing the MTF aB-tree} \label{sub_sec_MTF_space}
We now describe how to compress the aB-tree $\cT$ defined above, using Theorem \ref{thm_pat_ab_compress}, to support \textsc{Rank} 
(and hence \textsc{Dictionary}) queries under the MTF (followed by arithmetic) encoding, with respect to the desired space bound 
$H_0(\MTF(x))$. Let $O(t)$ be the desired query time.
Choose $B \geq 2$ such that $B \lg B = \frac{\epsilon \lg n}{\max(t|\Sigma|, \lg \lg n)}$
for some small $\epsilon > 0$. Let $r = B^t$. We divide $m$ into $n / r$ sub-arrays $A_1, A_2, \cdots, A_{n/r}$ of size $r$ and build an aB-tree over each sub-array. We show how to support \textsc{Dictionary} and \textsc{Rank} queries within each sub-array in time $O(\lg_B r) = O(t)$.


For each $j \in [n/r]$, we store the true MTF stack at the beginning of the sub-array $A_j$, the frequency of each character $c \in \Sigma$ in the prefix $x[1 : (j-1)r]$, and its index in memory.

Given a \textsc{Dictionary} query with index $i \in [n]$, the query algorithm determines the sub-array $A_j$ ($j = \lceil i/r \rceil$) containing $i$, initializes $S$ to the MTF stack $S_{(j-1)r}$ just before $A_j$, and performs the query algorithm described in Section \ref{sub_sec_local_decoding_MTF} on the aB-tree over $A_j$, with query index $i - (j-1)r$.

Similarly, given a \textsc{Rank} query $(c_\sigma,i) \in \Sigma \times [n]$, the query algorithm determines the sub-array $A_j$ containing $i$, reads $r' := rk_x(c_\sigma, (j-1)r)$, the rank of $c_\sigma$ in the prefix $x[1 : (j-1)r]$, and builds the permutation $\pi^* = \pi_{0,(j-1)r}$ corresponding to the MTF stack $S_{(j-1)r}$. Then, it performs the query algorithm described in Section \ref{sub_sec_rank_MTF} on the aB-tree over $A_j$, with query $\left(c_{\pi^*(\sigma)}, i - (j - 1)r\right) \in \Sigma \times [r]$. Finally, it adds $r'$ to this answer and returns the sum.


For a MTF character $\sigma \in \{0, 1, \cdots, |\Sigma|-1\}$, let $f_\sigma$ be the frequency of $\sigma$ in $m$. Following \cite{Pat}, we define a measure of ``entropy per character''. For $\sigma \in \{0, 1, \cdots, |\Sigma|-1\}$, we encode each occurrence of $\sigma$ in $m$ using $\lg \frac{n}{f_\sigma}$ bits, rounded up to the nearest multiple of $1/r$. We impose a \emph{zeroth-order entropy constraint} by augmenting each node $v$ with an additional value $\varphi_0(v)$, which is the sum of the entropy (suitably discretized, as described above) of the symbols in its subtree. We have
\[
	H_0(m) = \sum_{\sigma = 0}^{|\Sigma|-1} f_\sigma \lg \frac{n}{f_\sigma} = \sum_{i = 1}^n \lg \frac{n}{f_{m_i}} = \sum_{j = 1}^{n/r} \sum_{i \in A_j} \lg \frac{n}{f_{m_i}} = \sum_{j = 1}^{n / r} H_0(A_j),
\]
where $H_0(A_j)$ is the sum of entropy of the symbols in $A_j$. Note that the assigned entropy $\lg \frac{n}{f_\sigma}$ of each occurrence of a character $\sigma$ is a function of its frequency in the \emph{entire} array $m$ (not in $A_j$).

Let $\Phi_0$ be the alphabet of these values $\phi_0(v)$. As the (discretized) entropy of each occurrence of a character can attain one of $O(r \lg n)$ values and the subtree of each node has at most $r$ leaves, we have $|\Phi_0| = O(r^2 \lg n)$.

Thus, for each node $v$, we encode the vector of values $\varphi(v) = (\varphi_\pi(v), \varphi_{rk}(v), \varphi_0(v))$. Now, for a given value of $\varphi = (\varphi_\pi, \varphi_{rk}, \varphi_0)$, the number of arrays $A$ of length $r$ with $H_0(A) = \varphi_0$ is at most $2^{\varphi_0}$ by a packing argument. So, we have $\cN(r, \varphi) \leq \cN(r, \varphi_0) \leq 2^{\varphi_0}$, and hence we can apply Theorem \ref{thm_pat_ab_compress} to store an aB-tree of size $r$, having value $\varphi = (\varphi_\pi, \varphi_{rk}, \varphi_0)$ at the root, using $\varphi_0 + 2$ bits. Summing this space bound over all $n/r$ sub-arrays $A_j$, we get that the space required to store the aB-trees is at most $\sum_{j = 1}^{n / r} \left(H_0(A_j) + 2\right) = H_0(m) + 2n/r$ bits.

The additional space required to store the true MTF stack and the rank of each character $c \in \Sigma$ at the beginning of each sub-array $A_j, j \in [n/r]$, is at most $\frac{n}{r} \left(|\Sigma| \lg n + |\Sigma| \lg |\Sigma|\right)$.

Now we analyze the space required for the look-up tables. We have the alphabet size $|\Phi| = |\Phi_{\pi}| \cdot |\Phi_{rk}| \cdot |\Phi_0| \leq O(|\Sigma|! \cdot (r+1)^{|\Sigma|} \cdot r^2 \lg n)$ with $r = B^t$. So the look-up tables occupy (in words)
\[
O\left(|\Phi|^{B+1} + B \cdot |\Phi|^B\right) = 2^{O\left(B |\Sigma| \lg |\Sigma| + t |\Sigma| \cdot B \lg B + B \lg \lg n\right)} = 2^{O(\epsilon \lg n)} = n^{O(\epsilon)},
\]
where the penultimate equality follows by considering the value of $B$ in two cases:
\begin{itemize}
	\item If $t |\Sigma| > \lg \lg n$, then $B \lg B = \frac{\eps \lg n}{t |\Sigma|}$. So, $t |\Sigma| \cdot B \lg B = \eps \lg n$, and $B \lg \lg n \leq B \cdot t |\Sigma| \leq \eps \lg n$.
	\item Otherwise, $B \lg B = \frac{\eps \lg n}{\lg \lg n}$. So, $B \lg \lg n \leq \eps \lg n$, and $t |\Sigma| \cdot B \lg B \leq \lg \lg n \cdot B \lg B  = \eps \lg n$.
\end{itemize}
This space usage is negligible for small enough constant $\epsilon > 0$. However, as $B \geq 2$, the minimum redundancy is (ignoring poly$\log(n)$ terms)
\[
	O\left(|\Phi|^3\right) = O_{|\Sigma|}\left(r^{3 (|\Sigma| + 2)}\right) = O_{|\Sigma|}\left(r^{3 |\Sigma| + 6}\right)
\]
So, the redundancy is $O\left(\frac{n}{r} + r^{3 |\Sigma| + 6}\right)$. We balance the terms to get that the redundancy is $O\left(\max\left\{\frac{n}{r}, n^{1 - 1/(3|\Sigma| + 7)}\right\}\right)$. We use the assumption that  $|\Sigma| = O(1)$, and adjust $t$ by a constant factor, to get that the overall space requirement is 
\[
	s = H_0(m) + n \Big{/} \left(\frac{\lg n}{\max(t, \lg \lg n)}\right)^t + n^{1 - \Omega(1)}.
\]
This concludes the proof of Theorem \ref{thm_local_mtf_formal}.

\section{Succinct Rank Data Structure over $\HRLX$} \label{sec_rle}

In this section, we prove Theorem \ref{thm_exp_tradeoff_rk_L}, which is restated below:
\rank*

\paragraph{Setup and Notation.}
Recall that $L = \BWT(x) \in \Sigma^n$. Let $m = \MTF(L)$. Then $m$ is a string of length $n$ over the MTF alphabet $\{\mathbf{0}, \mathbf{1}, \mathbf{2}, \cdots, \mathbf{|\Sigma|-1}\}$ (boldface symbols indicate MTF characters). Let $\bar{m}$ be the string obtained from $m$ by replacing each run of $0$'s with a single character which represents its length. Thus, $\bar{m}$ is a string of length $\bar{N} \leq n$ over the expanded alphabet
$\bar{\Sigma} := [\mathbf{|\Sigma|-1}] \cup [n]$,
where $[\mathbf{|\Sigma|-1}] :=  \{\mathbf{1}, \mathbf{2}, \cdots, \mathbf{|\Sigma|-1}\}$. The information-theoretic minimum space required to encode $\bar{m}$ using a \emph{zeroth order prefix-free code} is
\[
	H_0(\bar{m}) = \sum_{\sigma \in \bar{\Sigma}} f_\sigma \lg \frac{\bar{N}}{f_\sigma},
\]
where $f_\sigma$ is the frequency of $\sigma$ in $\bar{m}$, for all $\sigma \in \bar{\Sigma}$. Consider any code which converts $x$ to $m = \MTF(L)$ using BWT followed by $\MTF$ encoding, and then compresses $m$ using Run-length Encoding of $0$-runs followed by prefix-free coding over the expanded alphabet $\bar{\Sigma} = [\mathbf{|\Sigma|-1}] \cup [n]$. This code requires at least $H_0(\bar{m})$ bits of space. In particular, we have $|\HRLX(L)| \geq H_0(\bar{m})$ by definition of the $\HRLX$ encoding.

We will build an aB-tree over a slightly modified encoding of $\bar{m}$, which is quite similar to the one defined in Section \ref{sec_MTF} but is succinct with respect to $H_0(\bar{m})$ (and hence with respect to $|\HRLX(L)|$).

Let $\epsilon \in (0,1)$ be a small constant. We divide each run of $0$'s of length $\ell_j > n^\epsilon$ in $m$ into $\big{\lceil}\frac{\ell_j}{n^\epsilon}\big{\rceil}$ runs of length at most $n^\epsilon$ each. We then replace each run of $0$'s by a single character which represents its length. Thus, we get a new string $m'$ of length $N \leq n$ over the alphabet
$\Sigma' := [\mathbf{|\Sigma|-1}] \cup [n^\epsilon]$.
This is done to minimize the space required for the additional look-up tables accompanying the aB-trees which is defined later. The following lemma ensures that this step increases the space usage of the aB-trees by at most an $\tilde{O}\left(n^{1 - \eps}\right)$ \emph{additive} term.

\begin{lemma} \label{lem_rle_run_divide_entropy}
	Let $\bar{m}$ and $m'$ be as defined above. Then
	\[
	H_0\left(m'\right) \leq H_0\left(\bar{m}\right) + O\left(n^{1-\epsilon} \lg n \right).
	\]
\end{lemma}

Intuitively, this lemma holds because the process of division of large runs introduces at most $n^{1 - \eps}$ additional symbols in $m'$ as compared to $\bar{m}$. Moreover, the \emph{relative} frequency of any character $\sigma \in \Sigma'$ only changes slightly, which allows us to bound the difference in the contribution of $\sigma$ to $H_0(\bar{m})$ and $H_0(m')$. We postpone the formal proof of this lemma to Section \ref{subsec_rle_run_divide_entropy_proof}.

\subsection{Succinct aB-tree over $m'$, and additional data structures}

Fix the branching factor $B \geq 2$ to be constant, and let $O(t')$ be the desired query time. Let $r = B^{t'}$. We divide $m'$ into $N/r$ sub-arrays $A'_1, A'_2, \cdots A'_{N/r}$ of length $r$, and build an aB-tree over each sub-array. We augment each node $v$ with a value $\varphi(v) = (\varphi_\pi(v), \varphi_{rk}(v), \varphi_0(v))$, where $\varphi_\pi(v) \in \cS_{|\Sigma|}$, $\varphi_{rk}(v) = \left(\varphi_{rk}(v, \sigma)\right)_{\sigma \in [\ell]} \in \{0, 1, \cdots, n^\eps \, r\}^{|\Sigma|}$, and $\varphi_0(v) \in [0, r \lg N]$. These augmented values have the same meaning as in Section \ref{sec_MTF}. We define these values formally below.

For each node $v$, we would like $\varphi_\pi(v) \in \cS_{|\Sigma|}$ to be the permutation induced by the MTF encoding on the sub-array of $m'$ (which corresponds to a contiguous sub-array of $\MTF(L)$) over which the subtree $\cT_v$ is built. Similarly, we would like $\varphi_{rk}(v, \sigma)$ to be the frequency of $c_\sigma$ in the sub-array rooted at $v$, assuming the MTF stack at the beginning of this sub-array is $S_0 = (c_1, c_2, \cdots, c_{|\Sigma|})$.

First, we define the augmented values at leaf nodes. Let $v$ be a leaf node corresponding to $m'_i \in \Sigma'$ for some $i \in [N]$. If $m'_i$ is a $\MTF$ symbol, i.e., $m'_i \in [\mathbf{|\Sigma|-1}]$, then we set $\varphi_\pi(v)$ and $\varphi_{rk}(v)$ exactly as defined in Section \ref{sec_MTF}. In particular, we define $\varphi_{rk}(v, \sigma^*) = 1$ for $\sigma^* = m'_i + 1$, and $\varphi_{rk}(v, \sigma) = 0$ for all $\sigma \neq \sigma^*$. If $m'_i$ corresponds to a run of $0$'s of length $\ell_j$ in $\MTF(L)$, then we define $\varphi_\pi(v) = \mathbf{Id}_{|\Sigma|}$ to be the identity permutation, $\varphi_{rk}(v, 1) = \ell_j$, and $\varphi_{rk}(v, \sigma) = 0$ for all $\sigma > 1$. Here, we use the fact that the $\MTF$ stack does not change within a run of $0$'s.

We now define the values $\varphi_\pi(v)$ and $\varphi_{rk}(v)$ at each internal node $v$ recursively in terms of the values at its $B$ children $v_1, v_2, \cdots, v_B$, as given by Equations \ref{eqn_defn_pi} and \ref{eqn_defn_rank} respectively.

Finally, we specify the entropy constraint $\varphi_0$. Recall that for $\sigma \in \Sigma' = [\mathbf{|\Sigma|-1}] \cup [n^\epsilon]$, $f_\sigma$ denotes the frequency of $\sigma$ in $m'$. For each $\sigma \in \Sigma'$, we encode each occurrence of $\sigma$ in $m'$ using $\lg \frac{N}{f_\sigma}$ bits, rounded up to the nearest multiple of $1/r$. We impose a \emph{zeroth-order entropy constraint} by augmenting each node $v$ with $\varphi_0(v)$, the sum of the entropy of the symbols in its subtree. By the same arguments as in Section \ref{sub_sec_MTF_space}, the space occupied by the aB-trees is at most $H_0(m') + 2N / r$.

Additionally, we store the following information, for each $j \in [N/r]$:
\begin{itemize}
	\item The true MTF stack $S_{(j-1)r}$ at the beginning of the sub-array $A'_j$.
	\item The frequency of each character $c \in \Sigma$ in the prefix $m'[1 : (j-1)r] = \left(A'_1, \cdots, A'_{j-1}\right)$.
	\item The index $i_j \in [n]$ of the character in $m$ corresponding to the first character $m'_{(j-1)r+1}$ of $A'_j$ (if $m'_{(j-1)r+1}$ represents a run in $m$, then we store the starting index of the run). 
	Let $T = \{i_j \in [n]\text{ } | \text{ } j \in [N/r]\}$ be the set of indices.
\end{itemize}
We also store the map $h : T \rightarrow [N/r]$, given by $h(i_j) = j$ for all $j \in [N/r]$. Finally, we build a predecessor data structure $D_{pred}$ over $T$. As there are at most $\frac{N}{r} \leq \frac{n}{r}$ keys from a universe of size $n$, there exists a data structure which can answer predecessor queries in time $O(t')$ using space $\frac{n}{r} \cdot r^{\Omega(1/t')} \cdot O(\lg n) = O\left(\frac{n \lg n}{B^{\Theta(t')}}\right)$ bits (for details, see \cite{PT}).

\subsection{Query algorithm} Let the query be $(c, i) \in \Sigma \times [n]$.
	\begin{itemize}
		\item Compute $i' = D_{pred}(i) \in T$, and index $j = h(i') \in [n/r]$ of the corresponding sub-array in $m'$. 
		\item Define and initialize the following variables:
		\begin{itemize}
			\item An MTF stack $S$, initialized to $S_{(j-1)r}$ (the true stack just before $A'_j$), as well as the corresponding permutation $\pi^* = \pi_{0, (j-1)r}$.
			\item A rank counter $rk$, initialized to the frequency of $c$ in the prefix $m'[1 : (j-1)r]$.
			\item A partial sum counter $PS$, initialized to $i'-1$. At any point, let $v$ be the last node visited by the query algorithm. Then $PS$ records the index in $m$ corresponding to the left-most node in the sub-array rooted at $v$ (for a run, we store its starting index).
		\end{itemize}
		\item Start from the root node of the aB-tree built over $A'_j$ and recursively perform the following for each node $v$ (with children $v_1, v_2, \cdots, v_B$) in the path (adaptively defined below), until a leaf node is reached:
		\begin{itemize}
			\item Let $\beta^* \in [B]$ be the largest index such that
			$PS + \sum\limits_{\beta = 1}^{\beta^* - 1} \sum\limits_{\sigma = 1}^{|\Sigma|} \varphi_{rk}(v_\beta, c_\sigma) \leq i$.
			\item Update $S$ and $rk$ as specified by \ref{eqn_mtf_stack_evolution} and \ref{eqn_rank_update_rule} respectively, with $\sigma$ replaced by $\pi^*(\sigma)$.
			\item Set $PS \leftarrow PS + \sum\limits_{\beta = 1}^{\beta^* - 1} \sum\limits_{\sigma = 1}^{|\Sigma|} \varphi_{rk}(v_\beta, c_\sigma)$.
			\item Recurse to $v_{\beta^*}$.
		\end{itemize}
		\item Let $m'_k$ be the character at the leaf node. If $m'_k$ represents a run of $0$'s, set $c' = S[1]$, the character at the top of the stack $S$. Otherwise, set $c' = S[m'_k + 1]$.
		\item If $c' = c$, set $rk \leftarrow rk + (i - PS)$.
		\item Return $rk$.  
	\end{itemize}
	Now we analyze the query time. The initial predecessor query and computation of sub-array index $j$ requires $O(t')$ time. Then, the algorithm spends $O(1)$ time per node in the aB-tree, which has depth $t'$. Hence, the overall query time is $O(t')$.

\subsection{Space Analysis}
Recall that we encoded an approximation of zeroth-order entropy constraint $\varphi_0$ as an augmented value in the aB-tree. Using Lemma \ref{lem_rle_run_divide_entropy} and arguments similar to those in Section \ref{sub_sec_MTF_space}, we have that the aB-trees occupy at most 
$H_0\left(\bar{m}\right) + \frac{2n}{r} + \tilde{O}(n^{1 - \eps})$ bits.
The additional data structures require $O\left(\frac{n \lg n \cdot |\Sigma|}{r}\right)$ bits of space.

Now we analyze the space required for the look-up tables. Let $\Phi = \Phi_\pi \times \Phi_{rk} \times \Phi_0$, where $\Phi_\pi$, $\Phi_{rk}$ and $\Phi_0$ are the alphabets over which the augmented values $\varphi_\pi$, $\varphi_{rk}$ and $\varphi_0$ are defined respectively. Then $|\Phi| = |\Phi_\pi| \cdot |\Phi_{rk}| \cdot |\Phi_0| \leq O\left(|\Sigma|! \cdot (n^\eps \cdot r)^{|\Sigma|} \cdot r^2 \lg n\right) = O_{|\Sigma|}\left(n^{\eps \cdot |\Sigma|} \cdot B^{t' (|\Sigma| + 2)} \cdot \lg n\right)$, as $r = B^{t'}$. So the look-up tables occupy (in words)
\[
O\left(|\Phi|^{B+1} + B \cdot |\Phi|^B\right) = O_{|\Sigma|}\left(n^{\eps \cdot |\Sigma| (B + 1)} \cdot B^{(B + 1) t' (|\Sigma| + 2)} \cdot \lg^{B+1} n\right) = n^{O(\epsilon + \delta)},
\]
as $B = \Theta(1)$ and $t' \leq \delta \lg n$ for a small constant $\delta > 0$. So, the look-up tables occupy negligible space for small enough $\epsilon, \delta > 0$. Thus, the overall space required (in bits) is at most
\[
	H_0(\bar{m}) + \frac{n \lg n \cdot |\Sigma|}{2^{\Omega(t')}} + n^{1 - \Omega(1)}.
\]
As $|\Sigma| = O(1)$ and $|\HRLX(L)| \geq H_0(\bar{m})$, we can adjust $t'$ by a constant factor to obtain Theorem \ref{thm_exp_tradeoff_rk_L}.
	
\subsection{Proof of Lemma \ref{lem_rle_run_divide_entropy}} \label{subsec_rle_run_divide_entropy_proof}

\begin{proof}
	
We assume $m$ contains at least one run of $0$'s of length exceeding $n^\eps$, since $H_0(m') = H_0(\bar{m})$ otherwise. Let $\kap$ and $\kap'$ be the number of characters representing $0$-runs in $\bar{m}$ and $m'$ respectively. Let $n'$ be the number of non-zero MTF characters in $m'$ (or $\bar{m}$). Then $N = n' + \kap'$, and $\bar{N} = n' + \kap$.

For $i \in [n]$, let $g_i$ be the number of runs of length $i$ in $\bar{m}$. For $i \in [n^\epsilon]$, let $\tilde{g}_i$ be the \emph{additional} number of runs of length $i$ introduced in $m'$ through this transformation. Let $\kappa_{sm}$ be the number of characters representing runs of length at most $n^\eps$ in $\bar{m}$.

The following facts are immediate:
	\begin{align}
	\sum_{i=1}^{n^\epsilon} g_i &= \kap_{sm} \leq \kap. \label{eqn_rle_1}\\
	\sum_{i = 1}^{n^\epsilon} \left(g_i + \tilde{g}_i\right) &= \kap' \leq \kap_{sm} + 2 n^{1 - \epsilon} \leq \kap + 2 n^{1 - \epsilon}. \label{eqn_rle_2}
	\end{align}
	We use these facts along with the Log-Sum Inequality to prove the lemma below.
	\begin{align*}
	&H_0\left(m'\right) - H_0\left(\bar{m}\right)\\
	&= \sum_{\sigma \in [\mathbf{|\Sigma|-1}]} f_\sigma \lg \frac{n' + \kap'}{f_\sigma} + \sum_{i = 1}^{n^\epsilon} \left(g_i + \tilde{g}_i\right) \lg \frac{n' + \kap'}{g_i + \tilde{g}_i} - \sum_{\sigma \in [\mathbf{|\Sigma|-1}]} f_\sigma \lg \frac{n' + \kap}{f_\sigma} - \sum_{i=1}^n g_i \lg \frac{n' + \kap}{g_i}\\
	&\leq \sum_{\sigma \in [\mathbf{|\Sigma|-1}]} f_\sigma \lg \frac{n' + \kap'}{n' + \kap} + \sum_{i = 1}^{n^\epsilon} g_i \lg \frac{n' + \kap'}{n' + \kap} + \sum_{i = 1}^{n^\epsilon} \tilde{g}_i \lg \frac{n' + \kap'}{g_i + \tilde{g}_i}\\
	&\leq (n' + \kap) \lg \frac{n' + \kap'}{n' + \kap} + \sum_{i = 1}^{n^\epsilon} \tilde{g}_i \lg \frac{ \sum_{i = 1}^{n^\epsilon} n' + \kap'}{\sum_{i = 1}^{n^\epsilon} \left(g_i + \tilde{g}_i\right)} \tag{Log-Sum Inequality}\\
	&\leq (n' + \kap) \lg \left(1 + \frac{2 n^{1-\epsilon}}{n' + \kap}\right) + (\kap' - \kap_{sm}) \lg \frac{n^\epsilon (n' + \kap')}{\kap'} \tag{$\sum_{i = 1}^{n^\epsilon} \tilde{g}_i = \kap' - \kap_{sm}$}\\
	&\leq 2n^{1-\epsilon} \left[\lg(e) + O(\lg n)\right] \tag{$\lg_2 (1 + x) \leq x \lg_2(e)$, $n' + \kap' \leq n$, $\kap' \geq 1$}\\
	&= O\left(n^{1-\epsilon} \lg n \right).
	\end{align*}
The last two inequalities above follow from \ref{eqn_rle_1} and \ref{eqn_rle_2}.
\end{proof}

\section{Reporting pattern occurrences} \label{sec_report_pattern_matching}

In this section, we prove the existence of a succinct data structure for reporting the positions of occurrences of a given pattern in a string. For $x \in \Sigma^n$ and a pattern $p \in \Sigma^\ell$, let $occ(p)$ be the number of occurrences of $p$ as a contiguous substring of $x$.

\begin{theorem} \label{thm_report_pattern_matching}
Fix a string $x \in \Sigma^n$. For any $t$, there is a succinct data structure that, given a pattern $p \in \Sigma^\ell$, reports the starting positions of the $occ(p)$ occurrences of $p$ in $x$ in time $O(t \cdot occ(p) + \ell \cdot \lg t)$, using at most 
\[
|\HRLX(\BWT(x))| + O\left(\frac{n \lg n \lg t}{t}\right) + n^{1 - \Omega(1)}
\]
bits of space, in the $w=\Theta(\lg n)$ word-RAM model.
\end{theorem}

For reporting queries, this is a quadratic improvement over the FM-Index \cite{FM}.

\begin{proof}
Let $t', \tilde{t} < t$ be parameters to be determined shortly. Let $D_{rk}$ be the data structure
given by Theorem \ref{thm_exp_tradeoff_rk_L} which supports \textsc{Rank} queries on $L$ in time $O(t')$. We divide the original string $x$ into $\lceil n/T\rceil $ blocks 
of size $T := O\left(\frac{t}{t' + \tilde{t}}\right)$. Let $S = \{(j-1)T + 1 \, | \, j \in [n/T]\}$ be the set of starting indices of blocks. Let $F_S$ be the set of indices in the first column $F$ of the BWT Matrix $\cM$ corresponding to indices in $S$. We store the map
$h : F_S \mapsto S$.
Moreover, we store a membership data structure on $[n]$, which given a query $i \in [n]$, answers \texttt{Yes} iff $i \in F_S$. This can be done using the data structure in \cite{Pat} which answers \textsc{Rank} queries over $\{0,1\}^n$\footnote{The bit-string $y$ is the indicator vector of $F_S$. For $i \in [n]$, we have $i \in F_S$ iff $rk_y(1, i) \neq rk_y(1, i-1)$.}
in time $O(\tilde{t})$ using space (in bits)
\[
\lg \binom{n}{n/T} + \frac{n}{\left(\lg n / \tilde{t}\right)^{\tilde{t}}} + \tilde{O}\left(n^{3/4}\right) \leq \frac{n}{T} \lg (eT) + \frac{n}{\left(\lg n / \tilde{t}\right)^{\tilde{t}}} + \tilde{O}\left(n^{3/4}\right). \tag{$\binom{n}{k} \leq \left(\frac{en}{k}\right)^k$}
\]
Our algorithm will follow the high-level approach to reporting pattern occurrences given in \cite{FM}, replacing each component data structure with the corresponding succinct data structure described above. We first use $D_{rk}$ to count the number of occurrences $occ(p)$ with $O(\ell \cdot |\Sigma|)$ \textsc{Rank} queries on $L$, which requires $O(\ell \, t')$ time. We observe that the algorithm for counting occurrences specified in \cite{FM} actually provides a contiguous set of rows $[R : R + occ(p) - 1] \subset [n]$ in the Burrows-Wheeler matrix $\cM$ which are prefixed by $p$. For each $i \in [R : R + occ(p) - 1]$, the algorithm starts from $i$ and performs $\alpha \leq T$ iterations of the LF-mapping algorithm\footnote{Note that the LF Mapping algorithm can equivalently be considered to be jumps in the first column $F$.} from 
Section \ref{sec_decoding_BWT}, until it reaches an index $i' \in F_S$ (which is verified using the membership data structure). Then it reports $h(i') + \alpha$.
 By Lemma \ref{lem_lf_inverse}, each step requires $O(|\Sigma|)$ \textsc{Rank} queries on $L$, each of which can be done using $\cD_{rk}$ in 
$O(t')$ time, hence the overall running time of the reporting phase is $O(T\cdot (t' + \tilde{t}) \cdot occ(p)) = O(t \cdot occ(p))$ by definition of $T$. We can assume $T \leq t = o(n)$, because otherwise we can decompress the entire string $x$ in time $O(n) = O(t)$. So $eT = o(n)$, and the total space required (in bits) is at most
\[
	|\HRLX(\BWT(x))| + \frac{n \lg n}{2^{t'}} + \frac{n}{(\lg n / \tilde{t})^{\tilde{t}}} + \frac{n}{T} \lg n + n^{1 - \Omega(1)}.
\]

In order to balance the redundancy terms in the above expression, we set $t' = \Theta(\lg T)$ . Using the fact that $\tilde{t} \leq t'$ if the first two redundancy terms are equal, we get $t = O(T \cdot (\tilde{t} + t')) = O(T \lg T)$, so $T = O(t / \lg t)$. Thus, we get a \emph{succinct} data structure for reporting the starting positions of the $occ(p)$ occurrences of a pattern $p \in \Sigma^\ell$ in $x$ in time $O(t \cdot occ(p) + \ell \cdot \lg t)$, using at most
\[
	|\HRLX(\BWT(x))| + O\left(\frac{n \lg n \lg t}{t}\right) + n^{1 - \Omega(1)}
\]
bits of space. This concludes the proof of Theorem \ref{thm_report_pattern_matching}.
\end{proof}

\section{Lower Bound for Symmetric Data Structures} \label{sec_LB}

In this section, we prove a cell-probe lower bound  
on the redundancy of any succinct data structure that locally decodes the BWT permutation $\Pi_x : \BWT(x) \mapsto x$, 
induced by a \emph{uniformly random} $n$-bit string $x\in\{0,1\}^n$.  
While it is somewhat unnatural to consider uniformly random (context-free) strings in BWT applications,  
we stress that our lower bound is more general and can yield nontrivial lower bounds for 
\emph{non-product} distributions $\mu$ on $\{0,1\}^n$ which 
satisfy the premise of our ``Entropy-Polarization" Lemma \ref{lem_ent_var} below, 
possibly with weaker parameters (we prove this lemma for the uniform distribution, but  
the proof can be generalized to distributions with ``sufficient block-wise independence").  

The lower bound we prove below applies to a somewhat stronger problem than Problem \ref{problem_1}, 
in which the data structure must decode \emph{both} forward and \emph{inverse} evaluations 
($\Pi_x(i), \Pi^{-1}_x(j)$) of the induced BWT permutation.  
The requirement that the data structure recovers $\Pi^{-1}_x(j)$, i.e., the 
position of $x_j$ in $L$, when decoding $x_j$, is very natural (and, when decoding the entire input $x$, is in fact without loss of generality). 
The ``symmetry" assumption, namely, the implicit assumption that any such data structure must also efficiently compute 
\emph{forward} queries $\Pi_x(i)$ 
mapping positions of $i \in L$ to their corresponding index $j \in X$, is less obvious, but is also a natural assumption 
given the sequential nature of the BWT decoding process (LF property, Lemma \ref{lem_LF}). 
Indeed, both the FM-index \cite{FM} and our data structure from Theorem \ref{thm_local_bwt} are essentially symmetric.\footnote{Indeed, we can 
achieve $q=O(t)$ by increasing the \emph{redundancy} $r$ by at most a factor of $2$, for storing an additional ``marking index" 
for the reverse permutation; see the first paragraph of Section \ref{sec_technical_overview}.}
We shall prove Theorem \ref{thm_LB}, restated below:
\thmLB*

When $q=\Theta(t)$, our result implies that obtaining an $r \ll n/t^2$ trade-off for Problem $1$ 
is generally impossible, hence Theorem \ref{thm_LB} provides an initial step in the lower 
bound study of Problem $1$.  

The data structure problem in Theorem \ref{thm_LB} is a variant of the \emph{succinct permutations} problem 
$\textsc{Perms}$ \cite{Golynski,MRRR}, 
in which the goal is to represent a random permutation $\Pi \in_R \cS_n$ succinctly using $ \lg n! + o(n \lg n)$ 
bits of space, supporting both forward and inverse evaluation queries ($\Pi(i)$ and $\Pi^{-1}(i)$), in query times $t,q$ respectively. 
Golynski \cite{Golynski}
proved a lower bound of $r \geq \Omega(n \lg n/tq)$ on the space redundancy of any such data structure, which 
applies whenever 
\begin{align}\label{cond_tq}
t,q \leq O\left(\frac{H(\Pi)}{n\cdot \lg \lg n}\right). 
\end{align}
When $\Pi\in_R \cS_n$ is a \emph{uniformly random} permutation, \eqref{cond_tq} implies that the bound holds 
for $t,q \leq O(\lg n/\lg\lg n)$. However, in the setting of Theorem \ref{thm_LB}, 
this result does not yield \emph{any} lower bound, since in our setting $H(\Pi_X)\leq n$ 
(as the BWT permutation of $X$ is determined by $X$ itself), hence \eqref{cond_tq} gives a trivial condition on the 
query times $t,q$. More precisely, the fact that $H(\Pi_X)\leq n$ implies 
that \emph{an average} query $i\in [n]$ only reveals $\frac{1}{n} \sum_{i=1}^n H\left(\Pi_X(i) \; | \Pi_X(i-1),\ldots , \Pi_X(1)\right) = O(1)$
bits of information on $X$, in sharp contrast to a \emph{uniformly random} permutation, where an average query reveals 
$\approx \lg n$ bits. This crucial issue completely dooms the cell-elimination argument in \cite{Golynski}, hence it fails to 
prove anything for Theorem \ref{thm_LB}.

In order to circumvent this problem, we first prove the following variant of Golynski's argument (Lemma 3.1 in \cite{Golynski}), 
which generalizes his lower bound 
on the \textsc{Perms} problem to arbitrary (i.e., nonuniform) random permutations $\Pi\sim \mu$, 
as long as there is a \emph{restricted subset} of queries $S \subseteq [n]$ with large enough entropy:

\begin{theorem}[Cell-Probe Lower Bound for Nonuniform \textsc{Perms}] 
\label{thm_LB_general}   
	Let $\Pi \sim \mu$ be a permutation chosen according to some distribution $\mu$ over $\cS_n$. Suppose that there exists a subset of coordinates $S = S(\Pi) \subseteq [n]$, $|S| = \gamma$, and $\alpha > 0, \eps \leq 1/2$ such that $H_\mu( \Pi_S | S, \Pi(S)) = \gamma\cdot \alpha \geq (1 - \eps) H_\mu(\Pi)$ bits. 
	Then any $0$-error succinct data structure for \textsc{Perms}$_n$ under $\mu$, in the cell-probe model with word size $w$,
	with respective query times $t, q$ satisfying
	\[
	t , q \leq \min\left\{2^{w/5}, \frac{1}{32}\cdot \frac{\alpha}{\lg w} \right\} \text{ and } tq \leq \delta \min \left(\frac{\alpha^2}{w \lg (en / \gamma)}, \frac{\alpha}{2 \eps \, w} \right)
	\]
	for some constant $\delta > 0$, must use $s \geq H_\mu(\Pi) + r$ bits of space in expectation, where
	\[
	r = \Omega\left(\frac{\alpha^2 \cdot \gamma}{w \cdot tq}\right)\; . 
	\]
\end{theorem}

Here, $ \Pi_S:= (\Pi(S_{i_1}),\Pi(S_{i_2}),\ldots, \Pi(S_{i_{s}}))$ 
denotes the projection of $\Pi$ to $S$, i.e., the 
\emph{ordered set} (vector) of evaluations of $\Pi$ on $S$, while $\Pi(S)$ denotes the image of $\Pi$ under $S$ (the unordered set), 
and $H_\mu(Z)$ is the Shannon entropy of $Z\sim \mu$.

\

Theorem \ref{thm_LB_general} implies that in order to prove a nontrivial bound in Theorem \ref{thm_LB}, it is enough 
to prove that the BWT-induced permutation $\Pi_X$ when applied on random $X$, has a (relatively) small subset of 
coordinates with near-maximal entropy (even though \emph{on average} it is constant). 
Indeed, we prove the following key lemma about the entropy distribution of the 
BWT permutation on random strings. Informally, it states that when $X$ is random, 
while the \emph{average} entropy of a coordinate (query) of $\Pi_X$ is indeed only $H_\mu(\Pi_X(i) | \Pi_X(<i)) = O(1)$,  
this random variable has a lot of \emph{variance}: A small subset ($\sim n/\lg n$) of coordinates have very high $(\sim \lg n)$ 
entropy, whereas the rest of the coordinates have $o(1)$ entropy on average.
This is the content of the next lemma:  

\begin{lemma}[Entropy Polarization of BWT] \label{lem_ent_var}
	Let $X \in_R \{0,1\}^n$, and $\mu$ be the distribution on $\cS_n$ induced by BWT on $X$. For any $\eps \geq \Omega(\lg \lg n / \lg n)$, with probability at least $1 - \tilde{O}\left(n^{- (1 - \eps/3)}\right)$, there exists a set $S \subset [n]$ of size $|S| = (1 - O(\eps)) n/ \lg n$, such that $H_\mu(\Pi_S | S, \Pi(S)) \geq n(1 - \eps)$.
\end{lemma}

We first prove Theorem \ref{thm_LB}, assuming Theorem \ref{thm_LB_general} and Lemma \ref{lem_ent_var}. We set $\eps = O(\lg \lg n / \lg n)$. Let $S \subset [n]$ be the set of size $\gamma = |S| = (1 - O(\eps)) n / \lg n$ obtained from Lemma \ref{lem_ent_var}. We invoke Theorem \ref{thm_LB_general} with the set $S$, $\alpha = (1 - O(\eps)) \lg n$, and $\mu$ being the distribution on $\cS_n$ induced by BWT on $X \in_R \{0,1\}^n$. As BWT is an invertible transformation on $\{0,1\}^n$ and $X$ is a uniformly random bit-string, $\Pi_X$ must be a uniformly random permutation over a subset of $\cS_n$ of size $2^n$. So, $H_\mu(\Pi) = n$. 

Let $\cD$ be any $0$-error data structure that computes 
$\Pi_X(i)$ and $\Pi^{-1}_X(j)$ for every $i,j\in [n]$ in time $t, q$ such that $t \, q \leq \delta \lg n/\lg \lg n$ for small enough $\delta > 0$. Then it is easy to see that with probability at least $1 - \tilde{O}\left(n^{-(1 - \eps)}\right)$, the conditions of Theorem \ref{thm_LB_general} are satisfied, and we get that $\cD$ must use
\[
s \geq \left(1 - \tilde{O}\left(n^{-(1 - \eps/3)}\right)\right) (H_\mu(\Pi) + r) = n + r - \tilde{O}\left(n^{\eps/3}\right) = n + r - o(r)
\] bits of space in expectation, where
\[
	r \geq \Omega\left(\frac{\alpha^2 \cdot \gamma}{w \cdot tq}\right) = \Omega\left(\frac{n}{tq}\right) \,.
\]
This concludes the proof of Theorem \ref{thm_LB}. Now we prove the Entropy Polarization Lemma \ref{lem_ent_var}, and then prove Theorem \ref{thm_LB_general}.

\subsection{Proof of Lemma \ref{lem_ent_var}} 

Let $X\in_R \{0,1\}^n$, and $L = \BWT(X)$. Let $\Pi_X : L \rightarrow X$ denote the BWT-permutation between 
indices of $X$ (i.e., $[n]$) and $L$. For convenience, we henceforth think of $\Pi_X$ as the 
permutation between $X$ and the \emph{first column} $F$ of the BWT matrix $\cM$ of $X$, i.e., $\Pi_X : F \rightarrow X$.
Note that these permutations are equivalent up to a rotation. We also denote $\Pi_X$ by $\Pi$, dropping the subscript $X$.

Recall that for any subset of coordinates $S\subseteq [n]$, $\Pi(S)$ denotes the image (unordered set 
of coordinates) of $S$ in $F$ under $\Pi$, and $\Pi_S$ denotes the \emph{ordered} set
(i.e., the projection of $\Pi$ to $S$).

Let $\ell = (1 + \eps) \lg n$, $\tilde{s} :=[\lfloor n/\ell \rfloor]$. For $i \in [\tilde{s}]$, let $Y_i = X[(i-1)\ell + 1 : i \ell]$ be the $i$th block of $X$ of length $\ell$. Let $J_i$ be the index of the row in the BWT matrix $\cM$ which starts with $Y_i$; note that $F[J_i]$ is the first character in $Y_i$. Let $\tilde{S} = \{J_i \, | \, i \in [\tilde{s}]\}$. So $|\tilde{S}| = \tilde{s} = n / ((1 + \eps) \lg n)$.

We first show the existence of a large subset of blocks $Y_i$ which are pairwise distinct, with high probability. The BWT ordering among the rows $J_i$ corresponding to the cyclic shifts starting with these blocks is consistent with the (unique) lexicographical ordering among the blocks themselves. We then show that the lexicographical ordering among these blocks is uniform. As this ordering is determined by the BWT permutation restricted to the corresponding rows $J_i$, the permutation restricted to these rows $J_i$ itself must have high entropy.

For $1 \leq i < j \leq \tilde{s}$, we say that there is a ``collision'' between blocks $Y_i$ and $Y_j$ if $Y_i = Y_j$, and define $Z_{i,j} \in \zo$ to be the indicator of this event. Let $Z = \sum_{1 \leq i < j \leq \tilde{s}} Z_{ij}$ be the number of collisions that occur among the disjoint blocks of length $\ell$.

\begin{claim} \label{clm-condition-good-event} 
	Let $\cE$ be the event that the number of collisions $Z \leq n^{1 - \eps} / \lg^2 n$. Then $\Pr[\cE] \geq 1 - \tilde{O}\left(n^{- (1 - \eps)} \right)$.
\end{claim}

\begin{proof}
	Fix $1 \leq i < j \leq \tilde{s}$. Then $Y_i$ and $Y_j$ are disjoint substrings of $X$ of length $\ell = (1 + \eps) \lg n$. As $X$ is a uniformly random string of length $n$, we have that $Y_i$ and $Y_j$ are independent and uniformly random strings of length $\ell$. Using this fact along with linearity of expectation, we have
	\[
	\E[Z] = \sum_{1 \leq i < j \leq \tilde{s}} \E[Z_{ij}] = \sum_{1 \leq i < j \leq \tilde{s}} \Pr[Y_i = Y_j] = \binom{\tilde{s}}{2} 2^{- \ell} = \binom{\tilde{s}}{2} \frac{1}{n^{1 + \eps}}.
	\]
	As $(1 - \eps) \frac{n}{\lg n} \leq \tilde{s} \leq \frac{n}{\lg n}$, we have
	\[
	\frac{n^{1 - \eps}}{4 \lg^2 n} \leq \E[Z] \leq \frac{n^{1 - \eps}}{2 \lg^2 n}.
	\]
	Fix $1 \leq i < j \leq \tilde{s}$. As $Z_{ij}$ is an indicator random variable, we have
	$\Var[Z_{ij}] \leq \Pr[Y_i = Y_j]$.
	It is easy to see that the random variables $Z_{ij}$ are pairwise independent. So, we have
	\[
	\Var[Z] = \sum_{1 \leq i < j \leq \tilde{s}} \Var[Z_{ij}] \leq \sum_{1 \leq i < j \leq \tilde{s}} \Pr[Y_i = Y_j] = \E[Z].
	\]
	Finally, we use Chebyshev's inequality and the bounds on $\E[Z]$ to conclude
	\[
	\Pr\left[Z > \frac{n^{1 - \eps}}{\lg^2 n}\right] \leq \Pr[|Z - \E[Z]| > \E[Z]] \leq \frac{\Var[Z]}{\E[Z]^2} \leq \frac{1}{\E[Z]} \leq \frac{4 \lg^2 n}{n^{1 - \eps}}.
	\]
\end{proof}

For the remainder of this proof, we assume that the event $\cE$ holds. Let $T = \{i \in [\tilde{s}] \, | \, Y_i \neq Y_j \, \forall \, j \neq i\}$ be the set of indices of blocks $Y_i$ which do not collide with any other block. Henceforth, we shall restrict our attention to the (unique) lexicographical order among blocks $Y_i, i \in [T]$. Let $s := |T|$, and $S := \{J_i \in [n]\, |\, i \in T\} \subset \tilde{S}$ be the corresponding set of indices in $F$. As $\cE$ holds, we have
\[
|S| = |T| = s \geq \tilde{s} - 2Z \geq \frac{n}{(1 + \eps) \lg n} - \frac{2 n^{1 - \eps}}{\lg^2 n} \geq \frac{n}{(1 + 2 \eps) \lg n} \geq \frac{n (1 - 2 \eps)}{\lg n}.
\]
Now, given $\cE$ and $T$, we define an ordering among the blocks $Y_i, i \in T$. For $j \in [s]$, we define $L_j = i^* \in [s]$ if $Y_{i^*}$ is the $j$th smallest string (lexicographically) among $Y_i, i \in T$. Finally, let $\cL = (L_1, \cdots, L_s)$. 
By abuse of notation, we can identify $\cL$ with a permutation on $s$ elements.

Note that, given $\cE$ and $T$, this order specified by $\cL$ is consistent with the relative ordering of the corresponding cyclic shifts of $X$ starting with $Y_i, i \in T$, in the BWT Matrix $\cM$. More precisely, we have $L_j = i^*$ if $J_{i^*}$ is the $j$th smallest row index among $S = \{J_i \, | \, i \in T\}$. Note that the definition of $\cL$ 
does not depend on $S$.

We will use the following fact about conditional entropy multiple times:
\begin{fact} \label{fact_entropy}
	Let $A, B$ be random variables in the same probability space. Then $H(A | B) \geq H(A) - H(B)$.
\end{fact}
\begin{proof}
	The proof is immediate using the chain rule:
	\[
	H(A) \leq H(A, B) = H(B) + H(A | B).
	\]
\end{proof}

The following claim allows us to remove the conditioning on the set $\Pi(S)$ at the cost of a negligible loss in entropy.
\begin{claim} \label{clm_lb_1}
	Let $S, \Pi(S)$ and $\Pi_S$ be as defined above. Then $H_\mu(\Pi_S | S, \Pi(S)) \geq H_\mu(\Pi_S | S) - O(s \cdot \lg (n/s))$.
\end{claim}
\begin{proof}
	We use Fact \ref{fact_entropy} with $A = \Pi_S | S$ and $B = \Pi(S)$ to get that $H_\mu(\Pi_S | S, \Pi(S)) \geq H_\mu(\Pi_S | S) - H_\mu(\Pi(S))$. Now, as $\Pi(S)$ is a $s$-size subset of the set $\{(i-1)\ell + 1 \, | \, i \in [\tilde{s}]\}$, we have that
	\[
	H_\mu(\Pi(S)) \leq \lg \binom{\tilde{s}}{s} = O\left(s \lg \frac{\bar{s}}{s}\right) \leq O\left(s \lg \frac{n}{s}\right) .
	\]
\end{proof}

We will invoke Claim 2 later with $s = \Theta(n / \lg n)$. Note that for any $\eps \geq \Omega(\lg \lg n / \lg n)$, we have that $s \lg (n/s) \leq O(s \lg \lg n) \leq \eps n$, so the loss in entropy is negligible.

\begin{claim} \label{clm_lb_2}
	Let $S, \cL, \Pi_S, \cE$ and $T$ be as defined above. Then $H_\mu(\Pi_S | S) \geq H(\cL | S, \cE, T)$.
\end{claim}
\begin{proof}
	It is enough to show that given $\cE$ and the set $S$, the permutation $\Pi_S$ (restricted to $S$) determines $T$ and the ordering $\cL$ among the blocks $Y_i, i \in T$. But this is true because, for any $j \in S$, the index $\Pi_S(j)$ specifies the position of $F[j]$ in $X$. In particular (as $j \in S$), it specifies the block index $i$ such that $F[j]$ is the first character of $Y_i$, i.e., $j = J_i$. Then $T$ is the set of these block indices. As this mapping is specified for all indices $j \in S$, the set $T$ and the ordering $\cL$ on $T$ are clearly determined by $\Pi_S$, given $S$. So, $H_\mu(\Pi_S | S) \geq H_\mu(\Pi_S | S, \cE) \geq H(\cL | S, \cE, T)$.
\end{proof}

As described earlier, given $\cE$ and $T$, $\cL$ is well-defined without reference to $S$. We now combine the previous claims to reduce the problem of lower-bounding $H_\mu(\Pi_S | S, \Pi(S))$ to that of lower-bounding $H(\cL | \cE, T)$, while losing negligible entropy.

\begin{claim} \label{clm_lb_reduce_pi_to_L}
	Let $\cE, T, \cL, S, \Pi(S)$ and $\Pi_S$ be as defined above. Then $H_\mu(\Pi_S | S, \Pi(S)) \geq H(\cL | \cE, T) - O(s \cdot \lg \lg n)$.
\end{claim}

\begin{proof}
	From Claims \ref{clm_lb_1}, \ref{clm_lb_2}, Fact \ref{fact_entropy} and the fact that $|\Pi(S)| = |S| = s = \Theta(n / \lg n)$, we have
	\begin{align*}
	H_\mu(\Pi_S | S, \Pi(S)) &\geq H_\mu(\Pi_S | S) - O\left(s \cdot \lg \frac{n}{s}\right)\\
	&\geq H(\cL | S, \cE, T) - O(s \cdot \lg \lg n)\\
	&\geq H(\cL | \cE, T) - O(s \cdot \lg \lg n)
	\end{align*}
	For the last inequality, we apply Fact \ref{fact_entropy} with $A = \cL | \cE,T$ and $B = S$, and upper bound $H_\mu(S)$ by $O(s \lg \lg n)$ using the same argument that was used to bound $H_\mu(\Pi(S))$.
\end{proof}

\begin{claim} \label{clm_L_high_entropy} 
	Let $\cE, \cL, T$ be as defined above. Then $H(\cL | \cE, T) = \lg (s!) \geq s \lg n - O(s \lg \lg n)$.
\end{claim}

\begin{proof}
	For $X \in \zo^n$ and permutation $\tau \in \cS_{\tilde{s}}$, define $X^\tau \in \zo^n$ to be the string obtained by shifting the entire $i$th substring $Y_i$ of $X$ to the $\tau(i)$th block, for all $i \in [\tilde{s}]$. Formally,  $X^\tau[(\tau(i)-1)\ell + 1: \tau(i) \ell] = X[(i-1) \ell + 1 : i \ell] = Y_i$. We observe that $X$ satisfies $\cE$ if and only if $X^\tau$ satisfies $\cE$. This is because $\cE$ is determined purely by inequalities among pairs of substrings $Y_i$, $Y_j$, which are kept intact by $\tau$ as it permutes entire blocks.
	
	Now, the event $\cE$ is clearly determined by $X$. So, for any fixed string $x \in \{0,1\}^n$, we have that $\Pr[\cE | X = x]$ is either $0$ or $1$. This implies that $\Pr[X = x | \cE] = \frac{\Pr[X = x]}{\Pr[\cE]}$ if $x$ satisfies $\cE$, and $\Pr[X = x | \cE] = 0$ otherwise.
	
	Let $x$ be a string which satisfies $\cE$. Let $\tau \in \cS_{\tilde{s}}$ be any permutation. Then
	\[
	\Pr[X = x \,|\, \cE] = \frac{\Pr[X = x]}{\Pr[\cE]} = 
	\frac{\Pr[X = x^\tau]}{\Pr[\cE]} = \Pr[X = x^\tau\, |\, \cE].
	\]
	The second equality follows from the fact that $X$ is uniform and the last equality follows because $x$ satisfies $\cE$ if and only if $x^\tau$ satisfies $\cE$. Similarly, for any permutation $\tau \in \cS_{\tilde{s}}$ and any string $x$ that does not satisfy $\cE$, we have
	$\Pr[X = x | \cE] = 0 = \Pr[X = x^\tau | \cE]$. This shows that the conditional distribution of $X$, given $\cE$, is still uniform over a subset of $\{0, 1\}^n$ of size $2^n (1 - o(1))$. 
	
	We show that, given $\cE$ and any fixed subset $T \subset [\tilde{s}]$ of size $s \geq n(1 - 2\eps) / \lg n$ such that the $\ell$-length substrings indexed by $T$ are all distinct, all $s!$ orderings $\cL$ on $T$ are equally likely. As $\cL$ is defined with respect to the lexicographical ordering among the blocks $Y_i, i \in T$, we only consider permutations $\tau \in \cS_{\tilde{s}}$ which fix indices outside $T$, i.e., $\tau(i) = i$ for all $i \in [\tilde{s}] \setminus T$. Fix any two distinct orderings $\cL_1, \cL_2 \in \cS_s$. For $i \in \{1,2\}$, let $\Lambda_i^T \subset \{0,1\}^n$ be the set of strings $x$ that satisfy event $\cE$ with respect to the set $T$ (i.e., the set of substrings $x[(i-1)\ell +1 : i \ell])$ are distinct for all $i \in T$), and give rise to the ordering $\cL_i$ on $T$. We show that $|\Lambda_1^T| = |\Lambda_2^T|$ by exhibiting a bijection $f : \Lambda_1^T \rightarrow \Lambda_2^T$.
	
	Let $\tau \in \cS_{\tilde{s}}$ be the unique permutation which converts $\cL_1$ to $\cL_2$, and fixes indices outside $T$. For $x \in \Lambda_1^T$, define $f(x) = x^\tau$. Then $f(x)$ satisfies $\cE$ as $x$ satisfies $\cE$. Moreover, as $x$ induces the order $\cL_1$ on $T$, $f(x)$ must induce the order $\cL_2$ on $T$. Hence, $f(x) \in \Lambda_2^T$. Clearly, this map is one-one. We observe that the inverse map $f^{-1} : \Lambda_2^T \rightarrow \Lambda_1^T$ is given by $f^{-1}(x) = x^{\tau^{-1}}$. So, $f$ is a bijection, and hence $|\Lambda_1^T| = |\Lambda_2^T|$. But any string in $\Lambda_1^T \cup \Lambda_2^T$ is equally probable under the distribution of $X$ conditioned on $\cE$. So, the probability of the induced ordering $\cL$ on $T$ being $\cL_1$ or $\cL_2$ is equal, i.e., $\Pr[\cL = \cL_1\, |\, \cE, T] = \Pr[\cL = \cL_2\, |\, \cE, T]$. As this is true for 
	all pairs of orderings on $T$, we have that $\cL$ is uniform on $s!$ orderings, given $\cE, T$. As $s = \Theta(n / \lg n)$, we have
	\[
	H(\cL | \cE, T) = \lg (s!) \geq s \lg s - O(s) \geq s \lg n - O(s \lg \lg n).
	\]
\end{proof}

We conclude the proof of Lemma \ref{lem_ent_var} by combining Claims \ref{clm_lb_reduce_pi_to_L} and \ref{clm_L_high_entropy}, and replacing $\eps$ by $\eps/3$:
\[
H_\mu(\Pi_S | S, \Pi(S)) \geq H(\cL | \cE, T) - O(s \cdot \lg \lg n) \geq s \lg n - O(s \cdot \lg \lg n) \geq n(1 - 2 \eps) - O(s \cdot \lg \lg n) \geq n(1 - 3 \eps).
\]

\subsection{Proof of Theorem \ref{thm_LB_general}}

\begin{proof}

	Let $\Pi \sim_\mu \cS_n$ and let $S=S(\Pi) \subseteq [n]$ be the subset of $|S|= \gamma$ 
	queries satisfying the premise of the theorem. 
	Let $D$ be a data structure for $\textsc{Perms}_n$ under $\mu$ which 
	correctly computes the (forward and inverse) answers 
	with query times $t,q$ respectively (satisfying the bounds in the theorem statement), 
	and 
	$s = h + r'$ words of space, where 
	$h = H_\mu(\Pi) / w$ and $r' = r/w$.
	
	The idea is to adapt Golynski's argument \cite{Golynski} to the restricted set of ``high entropy" queries 
	$(S,\Pi(S))$ and use $D$ in order to efficiently encode the answers of $\Pi_S$,  
	while the remaining answers ($\Pi_{\overline{S}}$) can be encoded in a standard way using 
	$\sim H_\mu(\Pi_{\bar{S}}| S, \Pi_S)$ extra bits (as this conditional entropy is very small by the premise).
	
	To this end, suppose Alice is given $\Pi \sim \mu$ and $S=S(\Pi)$. 
	Alice will first use $D_S$ to send $\Pi_S$, along with the sets $S, \Pi(S) \subset [n]$, to Bob. 
	The subsequent encoding argument proceeds in stages, where in each stage we delete one ``unpopular cell'' of $D_S$, i.e., a cell 
	which is accessed by \emph{few} forward and inverse queries in $S,\Pi(S)$, and protect some other cells from deletion 
	in future stages. The number of stages $z$ is set so that at the end, the number of remaining unprotected cells (i.e., the number 
	of cells which were not protected nor deleted in any stage) is at least $h/2$. 
	
	Following \cite{Golynski}, for any cell $l \in [h + r']$ we let $F_S(l)$ and $I_{\Pi(S)}(l)
	$ denote the subset of forward and inverse queries (in $S,\Pi(S)$ respectively) which probe $l$. 
	Fix a stage $k$. Let $C^S_k$ denote the \emph{remaining} cells at the end of stage $k$ 
	(i.e., the set of undeleted unprotected cells) probed by any query $i \in S \cup \Pi(S)$. 
	So $|C_0| = h + r'$ and $|C^S_z| \geq \frac{h}{2}$. 
	For a query $i\in S$, let $R_k(i)$ denote the set of remaining cells at stage $k$ probed by $i$.
	The average number of forward queries in $S$ that probe a particular cell is
	\[
	\frac{1}{|C^S_k|} \sum_{l \in C^S_k} |F_S(l)| = \frac{1}{|C^S_k|} \sum_{i \in S} |R_k(i)| \leq \frac{|S|\,t}{|C^S_k|} = \frac{\gamma \, t}{|C^S_k|}.
	\]
	As such, there are at most $|C^S_k|/3$ cells which are probed by at least $3 \frac{\gamma \, t}{|C^S_k|}$ forward queries in $S$. 
	By an analogous argument, there are at most $|C^S_k|/3$ cells which are probed by at least $3 \frac{\gamma \, q}{|C^S_k|}$ 
	inverse queries in $\Pi(S)$. Hence, we can find a cell $d_k \in C^S_k$ which is probed by at most $\beta
	:= 3 \frac{\gamma \, t}{|C^S_k|}$ 
	forward queries in $S$ and $\beta
	' := 3 \frac{\gamma \, q}{|C^S_k|}$ inverse queries in $\Pi(S)$. As $|C^S_k| \geq h/2$, we have that 
	$\beta
	\leq 6 \frac{\gamma \, t}{h}$ and $\beta' \leq 6 \frac{\gamma \, q}{h}$. We delete (the contents of) this cell $d_k$. Let 
	\[
	F_S'(d_k) := F_S(d_k) \cap \Pi^{-1}(I_{\Pi(S)}(d_k))
	\]
	denote the set of forward queries $i \in S$ such that both $i$ and the reciprocal $j := \Pi(i) \in \Pi(S)$ probe the deleted cell $d_k$ 
	(note that $F_S'(d_k)$ is well defined even when $D_S$ is adaptive, given the input $\Pi$).  
	Define $I_{\Pi(S)}'(d_k)$ analogously. We say that the queries in $F'_S(d_k)$ and $I'_{\Pi(S)}(d_k)$ are \emph{non-recoverable}.
	Note that, by definition, $\Pi$ defines a bijection between 
	$F'_S(d_k)$ and $I'_{\Pi(S)}$, i.e., 
	\begin{align}\label{prop_iv_remaining_cell}
	\Pi(F_S'(d_k)) = I_{\Pi(S)}'(d_k).
	\end{align}
	In order to ``preserve'' the information in $d_k$, Alice will send Bob an array $A_{d_k}$ of $\beta$ 
	entries, each of size $\lg \beta'$ bits, encoding the bijection 
	$\Pi_{F_S'(d_k)} : F_S'(d_k) \longleftrightarrow I_{\Pi(S)}'(d_k)$ 
	so that Bob can recover the answer to a query pair $(i, j)$ in the event that both $i$ and $j=\Pi(i)$ access 
	the deleted cell $d_k$.  Let $P(d_k) := \bigcup_{i\in \Pi(F_S(d_k)) \cup \Pi^{-1}(I_{\Pi(S)}(d_k))} R_k(i) \setminus \{d_k\}$ 
	denote the union of all remaining cells probed by ``reciprocal" queries to queries that probe $d_k$, 
	excluding $d_k$ itself.  
	To ensure the aforementioned  bijection 
	can be encoded ``locally" (i.e., using only $\beta\lg \beta'$ bits instead of $\beta \lg n$), 
	we \emph{protect} the cells in $P(d_k)$ from any future deletion. This is crucial, as it guarantees that 
	any query is associated with at most \emph{one} deleted cell. 
	
	The number of protected cells is at most $|P(d_k)| \leq t|I_{\Pi(S)}(d_k)| + q|F_S(d_k)| \leq t\beta' + q\beta \leq O(\gamma \, tq/h)$, 
	since $F_S(d_k) \cup I_{\Pi(S)}(d_k)$ contains at most $|F_S(d_k)| \leq \beta$ forward queries in $S$ and $|I_{\Pi(S)}(d_k)| \leq \beta'$ inverse queries in $\Pi(S)$. 
	This implies by direct calculation (\cite{Golynski}) that the number of stages $z$ that can be performed 
	while satisfying $|C^S_z| \geq h/2$, is 
	$$z = \Theta(h/q\beta) =  \Theta(h/t\beta'), $$ 
	since $q\beta= t\beta'$.
	Note that by \eqref{prop_iv_remaining_cell}, sending the bijection array $A_{d_k}$ can be done using 
	(at most) $\beta\lg \beta'$ bits, by storing, for each query in $F'_S(d_k)$, the corresponding index 
	from $I'_{\Pi(S)}(d_k)$ (where the sets are sorted in lexicographical order), since every query is associated 
	with at most one deleted cell. 
	Let $\cA$ denote the concatenation of all the arrays $\{A_{d_k}\}_{k\in[z]}$ of deleted cells in the $z$ stages 
	(occupying at most $z\beta\lg \beta'$ bits), 
	$\cL$ denote the locations of all deleted cells (which can be sent using at most 
	$\lg {h+r' \choose z} \leq z\lg(O(h)/z) = z (\lg (\beta q) + O(1))$ bits, assuming $r' = O(h)$), and $\cR$ denote the contents of remaining cells, 
	which occupy $$h+r'-z$$ words. 
	Alice sends Bob $M := (\cA,\cR,\cL)$, along with the explicit sets $S,\Pi(S)$. Alice also sends Bob the answers to the forward queries outside $S$, conditioned on $S$ and the answers $\Pi_S$ on $S$, using at most $H_\mu(\Pi_{\bar{S}} | S, \Pi_S) + 1$ bits using standard Huffman coding. Let this message be $\cP$.
	
	Assume w.l.o.g that $t \leq q$, hence $\beta\leq \beta'$.
	Recalling that $\beta' \leq 6 \frac{\gamma \, q}{h}$, the premises of the lemma 
	\[ q \leq \min\left\{2^{w/5}, \frac{1}{32}\cdot \frac{\alpha}{\lg w} \right\}, \text{ and \;\;}  h \geq \frac{1}{w} H_\mu(\Pi_S | S, \Pi(S)) = \frac{\gamma \cdot \alpha}{w} \]    
	imply that $\beta' \leq \frac{6 \, q\,\gamma \,w}{\gamma \,\alpha} \leq \frac{6\,\alpha\, w}{32\,\alpha\, \lg w} < \frac{w}{5 \lg w}$, hence 
	the total cost of sending $\cA$ is at most 
	$$z\beta\lg \beta' \leq z\beta'\lg \beta' \leq z \frac{w}{5 \lg w} \cdot \lg w = zw/5$$ bits. 
	Since $\lg q \leq w/5$ by assumption, the cost of sending $\cL$ is at most 
	$z(\lg (\beta q) + O(1)) \leq zw/5 + z\lg w + O(z) = zw/5 +o(zw)$ bits.  
	Sending the sets $S,\Pi(S)$ 
	can be done using at most $2|S| \lg (en/|S|) = 2 \gamma \lg(en/\gamma)$ bits. Now, as $tq \leq \frac{\delta \cdot \alpha^2}{w \lg (en / \gamma)}$ for small enough $\delta > 0$, we have
	\[
	zw = \Theta\left(\frac{hw}{q \beta}\right) \geq \Omega\left(\frac{h^2 \,w}{\gamma \,t \, q}\right) \geq \Omega\left(\frac{\gamma \, \alpha^2}{w \, t \, q}\right) \geq 10 \gamma \lg (en/\gamma).
	\]
	So, Alice can send $S, \Pi(S)$ using at most $zw/5$ bits. Finally, the message $\cP$ requires (up to $1$ bit)
	\[
	H_\mu(\Pi_{\bar{S}} | S, \Pi_S) \leq \eps H_\mu(\Pi) \leq 2 \eps \alpha \gamma  \leq \frac{\delta \, \gamma \, \alpha^2}{w \, tq} \leq \frac{zw}{5}
	\]
	bits in expectation with respect to $\mu$. Thus, together with the cost of $\cR$ and $\cP$, Alice's total expected message size is at most 
	$(h+r'-z)w + 4zw/5$ bits, or $h+r'-z + 4z/5 = h + r' - z/5$ words (up to $o(z)$ terms). But the minimum expected length of Alice's message must be $h$ words. 
	This implies that the redundancy \emph{in bits} must be at least \[ r = r' w \geq \Omega(zw)\] 
	in expectation, which, on substituting $z = \Theta(h/q\beta)$, yields the claimed lower bound on $r$, 
	assuming Bob can recover $\Pi$ from Alice's message, which we argue below. 
	
	To decode $\Pi_S := \{\Pi(i)\}_{i\in S}$ given $S$ and Alice's message $M$, Bob proceeds as follows: 
	He first fills an empty memory of $h+r'$ words with the contents of $\cR$ he receives from Alice, leaving 
	all deleted cells empty. Let $A,A'$ denote the forward (resp. inverse) query algorithms of $D_S$ (note that 
	$A,A'$ are defined for \emph{all} queries in $[n]$). Bob simulates the forward query algorithm $A(i)$ on all forward 
	queries in $S$ and the inverse query algorithm $A'(j)$ on all inverse queries $j\in \Pi(S)$. If a query $i \in S$ \emph{fails} (i.e., probes some deleted cell) but $A'$ on some inverse query $j \in \Pi(S)$ does not fail and returns $A'(j) = i$, he can infer the answer $\Pi(i) = j$. Similarly, he can infer the answer to an inverse query which fails, if its corresponding forward query does not fail. So, we focus on the non-recoverable queries in $S$ and $\Pi(S)$.
	
	If a query $i\in S$ \emph{fails}, he finds the \emph{first} deleted 
	cell $d$ probed by $A(i)$, and lexicographically lists \emph{all} the queries in $F'_S(d)$ for which $d$ is the 
	\emph{first} deleted cell. 
	Similarly,
	he lexicographically lists the set of all inverse queries in $I'_{\Pi(S)}(d)$ 
	whose first deleted cell is $d$. 
	He then uses the array $A_{d}$, which stores the bijection between the non-recoverable forward queries in 
	$F'_S(d)$ and the non-recoverable inverse queries in $I'_{\Pi(S)}(d)$, to answer these non-recoverable forward queries. In an analogous manner, he answers the non-recoverable inverse queries.
	Note that we crucially use the fact that each query accesses at most one deleted cell.
	
	Finally, Bob uses $\cP$ along with the answers to queries in $S$ and $\Pi(S)$ to answer all forward queries in $\bar{S}$. At this point, since he knows the answers to all forward queries, he can recover $\Pi$, and answer all inverse queries in $\bar{\Pi}(S)$ as well.
\end{proof}

\bibliographystyle{plain}
\bibliography{localbwt_online}

\begin{appendix}

	
	
	\section{Proof of Fact \ref{lem_LF} (LF Mapping Property)} \label{sec_LF}
	We prove the first equality in Fact \ref{lem_LF}. The intuition is that the order among distinct occurrences of a character $c$ in both $F$ and $L$ are decided by the character's \emph{right-context} in $x$, and so they must be the same.
	
	\begin{proof}
	Fix $c \in \Sigma$. Consider any two distinct occurrences of $c$ in $x$, at positions $k_1, k_2 \in [n]$ respectively. For $\alpha \in \{1,2\}$, let $i_\alpha, j_\alpha \in [n]$ be indices such that the BWT maps $x[k_\alpha]$ to position $i_\alpha$ in $L$ and position $j_\alpha$ in $F$, i.e., $x[k_\alpha] = L[i_\alpha] = F[j_\alpha] = c$. Then it suffices to prove that $i_1 < i_2$ if and only if $j_1 < j_2$, because the ordering among all occurrences of $c$ is determined by the relative ordering among all pairs of occurrences of $c$.
	
	For $\alpha \in \{1,2\}$, let $y_\alpha = x[k_\alpha + 1, k_\alpha + 2, \cdots, n, 1, \cdots, k_\alpha - 1]$ denote the \emph{right-context} of $x[k_\alpha]$ in $x$, which corresponds to the cyclic shift of $x$ by $\alpha$ positions (and then excluding $x[k_\alpha]$). For $\beta \in [n]$, let $\cM[\beta]$ denote row $\beta$ of the BWT matrix $\cM$. Then it is easy to see that $\cM[i_\alpha] = (y_\alpha, c)$ and $\cM[j_\alpha] = (c, y_\alpha)$, for $\alpha \in \{1,2\}$. We write $z_1 \prec z_2$ below to mean that $z_1$ is smaller than $z_2$ according to the lexicographical order on $\Sigma$.
	
	Assume $i_1 < i_2$. From this assumption and the fact that the rows of $\cM$ are sorted lexicographically, we have $(y_1, c) \prec (y_2, c)$. But this implies that $y_1 \prec y_2$ lexicographically, as both strings are of equal length and end with $c$. So, we have 
	\[
		\cM[j_1] = (c, y_1) \prec (c, y_2) = \cM[j_2].
	\]
	We conclude that $j_1 < j_2$ if $i_1 < i_2$. Clearly, the converse also holds. Thus, the relative ordering between any two occurrences of $c$ is the same in $F$ and $L$. This concludes the proof of Fact \ref{lem_LF}.
	\end{proof}

	\section{The $\HRLX$ Benchmark and Comparison to Other Compressors} \label{sec_app_RLX}
	A theoretical justification of the $\HRLX$ space benchmark was first given by \cite{Manzini99,FM}, where it was proved that 
	$\HRLX(\BWT(x))$ approaches the infinite-order empirical entropy of $x$
	(even under the weaker version of~\cite{FM} where the final arithmetic coding stage (3) is excluded), namely, that for \emph{any} 
	constant $k\geq 0$, \[ |\HRLX(\BWT(x))| \leq 5\cdot H_k(x) +O(\lg n).\]
	Several other compression methods were subsequently proposed  for compressing $\BWT(x)$ 
	(e.g., \cite{DC00,FGM_wavelet,FGMS_Booster}), some of which achieving better (essentially optimal) worst-case 
	theoretical guarantees with respect to $H_k(x)$, albeit at the price of an $\Omega(n/\lg n)$ or even $\Omega(n)$ 
	additive factor, which becomes the dominant term in the interesting regime of compressed text indexing \cite{KM99,FT}.  
	The same caveat holds for other entropy-coding methods such as LZ77, LZ78 and PPMC~\cite{LZ78, LZ78, Moffat90}, 
	confirming experimental results which demonstrated the superiority of BWT-based text compression \cite{Effros02,Kaplan,Manzini99}.  
	Indeed, Kaplan et. al \cite{Kaplan,Kaplan_lower_bound} observed that, despite failing to converge \emph{exactly} to $H_k$, 
	distance-based and MTF compression of BWT such as $\HRLX$ tend to outperform other entropy coding methods (especially in the 
	presence of long-term correlations as in English text).  
	$\HRLX$ is the basis of the \texttt{bzip2} program \cite{bzip2}. 
	For further elaboration 
	we refer the reader to \cite{Manzini99,FGMS_Booster}.

\end{appendix}

\end{document}